\documentclass[sigplan,10pt]{acmart}

\newif\ifincludeappendix
\includeappendixtrue 
\ifincludeappendix\fi 

\usepackage{doi}
\usepackage{tikz}
\usetikzlibrary {arrows.meta}

\usepackage{xspace}
\newcommand{\algname}{Eg-walker\xspace}

\usepackage[frozencache]{minted}

\copyrightyear{2025}
\acmYear{2025}
\setcopyright{cc}
\acmConference[EuroSys '25]{Twentieth European Conference on Computer Systems}{March 30-April 3, 2025}{Rotterdam, Netherlands}
\acmBooktitle{Twentieth European Conference on Computer Systems (EuroSys '25), March 30-April 3, 2025, Rotterdam, Netherlands}
\acmDOI{10.1145/3689031.3696076}
\acmISBN{979-8-4007-1196-1/25/03}
\acmCodeLink{https://github.com/josephg/egwalker-paper}

\begin{document}
\def\sectionautorefname{Section}%
\def\subsectionautorefname{Section}%
\def\subsubsectionautorefname{Section}%
\def\listingautorefname{Listing}%

\title[Collaborative Text Editing with Eg-walker: Better, Faster, Smaller]{Collaborative Text Editing with Eg-walker:\\ Better, Faster, Smaller}
\author{Joseph Gentle}
\email{me@josephg.com}
\orcid{0009-0005-9322-1973}
\affiliation{%
  \institution{Independent}
  \city{Melbourne}
  \country{Australia}
}

\author{Martin Kleppmann}
\email{martin.kleppmann@cst.cam.ac.uk}
\orcid{0000-0001-7252-6958}
\affiliation{%
  \institution{University of Cambridge}
  \city{Cambridge}
  \country{United Kingdom}}

\begin{abstract}
  Collaborative text editing algorithms allow several users to concurrently modify a text file, and automatically merge concurrent edits into a consistent state.
  Existing algorithms fall in two categories: Operational Transformation (OT) algorithms are slow to merge files that have diverged substantially due to offline editing; CRDTs are slow to load and consume a lot of memory.
  We introduce \algname, a collaboration algorithm for text that avoids these weaknesses.
  Compared to existing CRDTs, it consumes an order of magnitude less memory in the steady state, and loading a document from disk is orders of magnitude faster.
  Compared to OT, merging long-running branches is orders of magnitude faster.
  In the worst case, the merging performance of \algname is comparable with existing CRDT algorithms.
  \algname can be used everywhere CRDTs are used, including peer-to-peer systems without a central server.
  By offering performance that is competitive with centralised algorithms, our result paves the way towards the widespread adoption of peer-to-peer collaboration software.
\end{abstract}


\maketitle

\section{Introduction}\label{introduction}

Real-time collaboration has become an essential feature for many types of software, including document editors such as Google Docs, Microsoft Word, or Overleaf, and graphics software such as Figma.
In such software, each user's device locally maintains a copy of the shared file (e.g. in a tab of their web browser).
A user's edits are immediately applied to their own local copy, without waiting for a network round-trip, so that the user interface is responsive regardless of network latency.
Different users may therefore make edits concurrently; the software must merge such concurrent edits in a way that maintains the integrity of the document, and ensures that all devices converge to the same state.

For example, in \autoref{two-inserts}, two users initially have the same document ``Helo''.
User 1 inserts a second letter ``l'' at index 3, while concurrently user 2 inserts an exclamation mark at index 4.
When user 2 receives the operation $\mathit{Insert}(3, \text{``l''})$ they can apply it to obtain ``Hello!'', but when user 1 receives $\mathit{Insert}(4, \text{``!''})$ they cannot apply that operation as-is, since that would result in the state ``Hell!o'', which would be inconsistent with the other user's state and the intended insertion position.
Due to the concurrent insertion at an earlier index, user 1 must insert the exclamation mark at index 5.

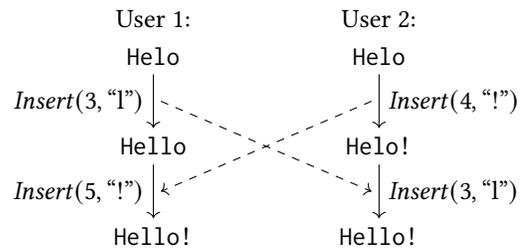
\begin{figure}
  \centering
  \begin{tikzpicture}
    \node at (0,2.9) {User 1:};
    \node at (3,2.9) {User 2:};
    \node (left1) at (0,2.4) {\texttt{Helo}};
    \node (left2) at (0,1.2) {\texttt{Hello}};
    \node (left3) at (0,0.0) {\texttt{Hello!}};
    \node (right1) at (3,2.4) {\texttt{Helo}};
    \node (right2) at (3,1.2) {\texttt{Helo!}};
    \node (right3) at (3,0.0) {\texttt{Hello!}};
    \draw [->] (left1) -- node [left] {$\mathit{Insert}(3, \text{``l''})$} (left2);
    \draw [->] (right1) -- node [right] {$\mathit{Insert}(4, \text{``!''})$} (right2);
    \draw [->] (left2) -- node [left] {$\mathit{Insert}(5, \text{``!''})$} (left3);
    \draw [->] (right2) -- node [right] {$\mathit{Insert}(3, \text{``l''})$} (right3);
    \draw [->,dashed] (0.1,1.8) -- (2.9,0.6);
    \draw [->,dashed] (2.9,1.8) -- (0.1,0.6);
  \end{tikzpicture}
  \caption{Two concurrent insertions into a text document.}
  \label{two-inserts}
\end{figure}

One way of solving this problem is to use \emph{Operational Transformation} (OT): when user 1 receives $\mathit{Insert}(4, \text{``!''})$ that operation is transformed with regard to the concurrent insertion at index 3, which increments the index at which the exclamation mark is inserted.
OT is an old and widely-used technique: it was introduced in 1989 \cite{Ellis1989}, and the OT algorithm Jupiter \cite{Nichols1995} is used in Google Docs \cite{DayRichter2010}.

OT is simple and fast in the case of \autoref{two-inserts}, where each user performed only one operation since the last version they had in common.
In general, if the users each performed $n$ operations since their last common version, merging their states using OT has a cost of at least $O(n^2)$, since each of one user's operations must be transformed with respect to all of the other user's operations.
Some OT algorithms' merge complexity is cubic or even slower \cite{Li2006,Roh2011RGA,Sun2020OT}.
This is acceptable for online collaboration where $n$ is typically small, but for larger $n$ an algorithm with complexity $O(n^2)$ can become impracticably slow.
In \autoref{benchmarking} we show a real-life example document that takes one hour to merge using OT.

Larger divergence occurs if users may edit a document offline, or if the software supports explicit branching and merging workflows.
In version control systems like Git, used mostly for software development, offline working and explicit branching are already the norm.
Recent research indicates that such workflows would also be valuable for writing prose \cite{Upwelling,Patchwork}, but OT-based collaborative editors struggle to offer such features because of the cost of merging substantially diverged branches.

\emph{Conflict-free Replicated Data Types} (CRDTs) have been proposed as an alternative to OT.
The first CRDT for collaborative text editing appeared in 2006 \cite{Oster2006WOOT}, and over a dozen text CRDTs have been published since \cite{crdt-papers}.
These algorithms work by maintaining additional metadata: they give each character a unique identifier, and use those IDs instead of integer indexes to identify the position of insertions and deletions.
This avoids having to transform operations, since IDs are not affected by concurrent operations.

Unfortunately, these IDs need to be loaded from disk when a document is opened, and held in memory while a document is being edited.
Some CRDT algorithms also need to retain IDs of deleted characters (\emph{tombstones}).
Early CRDT algorithms were very inefficient, using hundreds of bytes of memory for each character of text, making them impractical for long documents.
Recent CRDT implementations have reduced this overhead considerably, but as we show in \autoref{benchmarking}, even the best CRDTs available today use more than 10 times as much memory as OT to view and edit a document.
For this reason, popular apps like Google Docs~\cite{DayRichter2010}, Microsoft Office, and Overleaf~\cite{overleaf-ot} use OT.
Existing algorithms therefore present a trade-off: either use OT and accept that offline editing and long-running branches are slow, or pick a CRDT and accept a much higher memory use.

In this paper we propose \emph{Event Graph Walker} (\algname), a collaborative editing algorithm that overcomes this trade-off.
Like OT, \algname uses integer indexes to identify insertion and deletion positions, and transforms those indexes to merge concurrent operations.
When two users concurrently perform $n$ operations each, \algname can merge them at a cost of $O(n \log n)$, much faster than OT's cost of $O(n^2)$ or worse.
The example document that takes 1 hour to merge using OT is merged in just 24~ms using \algname (\autoref{chart-remote}).

\algname merges concurrent edits using a CRDT algorithm we designed.
Unlike existing algorithms, we invoke the CRDT only to perform merges of concurrent operations, and we discard its state as soon as the merge is complete.
We never write the CRDT state to disk and never send it over the network.
While a document is being edited, we only hold the document text in memory, but no CRDT metadata.
Most of the time, \algname therefore uses 1–2 orders of magnitude less memory than the best CRDTs.
During merging, when \algname temporarily uses more memory, its peak memory use is comparable to the best known CRDT implementations.

\algname assumes no central server, so it can be used over a peer-to-peer network.
Although all existing CRDTs and a few OT algorithms can be used peer-to-peer, most of them have poor performance compared to the centralised OT commonly used in production software.
In contrast, \algname's performance matches or surpasses that of centralised algorithms.
It therefore paves the way towards more collaboration software working peer-to-peer, for example in environments where co-located devices can communicate via local radio links, but not reach the Internet or any cloud services.
This setting is important e.g. for devices onboard the same aircraft~\cite{ditto-aircraft}, in a military context~\cite{ditto-military}, or for scientists conducting fieldwork in remote locations~\cite{antarctica}.

This paper focuses on collaborative editing of plain text files.
We believe that our approach can be generalised to other file types such as rich text, spreadsheets, graphics, presentations, CAD drawings, and more in the future.
More generally, \algname provides a framework for efficient coordination-free distributed systems, in which nodes can always make progress independently, but converge eventually \cite{Hellerstein2010}.

This paper makes the following contributions:

\begin{itemize}
\item We introduce \algname, a hybrid CRDT/OT algorithm for text that is faster and has a vastly smaller memory footprint than existing CRDTs (\autoref{algorithm}).
\item Since there is no established benchmark for collaborative text editing, we are also publishing a suite of editing traces of text files for benchmarking. They are derived from real documents and demonstrate various patterns of sequential and concurrent editing.
\item In \autoref{benchmarking} we use those editing traces to evaluate the performance of our implementation of \algname, comparing it to selected CRDTs and an OT implementation. We measure CPU time to load a document, CPU time to merge edits from a remote replica, memory usage, and file size. \algname improves the state of the art by orders of magnitude in the best cases, and is only slightly slower in the worst cases.
\item We prove the correctness of \algname in \ifincludeappendix\autoref{proofs}.\else the extended version of this paper~\cite{extended-version}.\fi
\end{itemize}

\section{Background}

We consider a collaborative plain text editor whose state is a linear sequence of characters, which may be edited by inserting or deleting characters at any position.
Such an edit is captured as an \emph{operation}; the operation $\mathit{Insert}(i, c)$ inserts character $c$ at index $i$, and $\mathit{Delete}(i)$ deletes the character at index $i$ (indexes are zero-based).
Our implementation compresses runs of consecutive insertions or deletions, but for simplicity we describe the algorithm in terms of single-character operations.

\subsection{System model}

Each device on which a user edits a document is a \emph{replica}, and each replica stores the full editing history of the document.
When a user makes an insertion or deletion, that operation is immediately applied to the user's local replica, and then asynchronously sent over the network to any other replicas that have a copy of the same document.
Users can also edit their local copy while offline; the corresponding operations are then enqueued and sent when the device is next online.

Our algorithm assumes a reliable broadcast protocol that detects and retransmits lost messages, but makes no other assumptions about the network.
For example, a relay server could store and forward messages from one replica to the others, or replicas could use a peer-to-peer gossip protocol.
We make no timing assumptions and tolerate arbitrary network delay, but we assume replicas are non-Byzantine.

Our algorithm ensures \emph{convergence}: any two replicas that have seen the same operations have the same document state (i.e., a text consisting of the same sequence of characters), even if the operations arrived in a different order at each replica.
If the underlying broadcast protocol ensures that every non-crashed replica eventually receives every operation, the algorithm achieves \emph{strong eventual consistency} \cite{Shapiro2011}.

\subsection{Event graphs}\label{event-graphs}

We represent the editing history of a document as an \emph{event graph}: a directed acyclic graph (DAG) in which every node is an \emph{event} consisting of an operation (insert/delete a character), a unique ID, and the set of IDs of its \emph{parent events}.
When $a$ is a \emph{parent} of $b$, we also say $b$ is a \emph{child} of $a$, and the graph contains an edge from $a$ to $b$.
We construct events such that the graph is transitively reduced (i.e., it contains no redundant edges).
When there is a directed path from $a$ to $b$ we say that $a$ \emph{happened before} $b$, and write $a \rightarrow b$ as per Lamport \cite{Lamport1978}.
The $\rightarrow$ relation is a strict partial order.
We say that events $a$ and $b$ are \emph{concurrent}, written $a \parallel b$, if both events are in the graph, $a \neq b$, and neither happened before the other: $a \not\rightarrow b \wedge b \not\rightarrow a$.

The \emph{frontier} is the set of events with no children.
Whenever a user performs an operation, a new event containing that operation is added to the graph, and the previous frontier in the replica's local copy of the graph becomes the new event's parents.
The new event is then broadcast over the network, and each replica adds it to its copy of the graph.
If any parents are missing (i.e., a parent ID in the event does not resolve to a known event), the replica waits for them to arrive before adding them to the graph; the result is a simple causal broadcast protocol \cite{Birman1991,Cachin2011}.
Two replicas can merge their event graphs by taking the union of their sets of events.
Events in the graph are immutable; they always represents the operation as originally generated, and not as a result of any transformation.
The graph grows monotonically (we never remove events), and a new event is always a child of existing events (we never add a parent to an existing event).

\begin{figure}
  \begin{tikzpicture}[every node/.style={inner ysep=1pt}]
    \node (char1) at (0,2.8) {$e_1: \mathit{Insert}(0, \text{``H''})$};
    \node (char2) at (0,2.1) {$e_2: \mathit{Insert}(1, \text{``e''})$};
    \node (char3) at (0,1.4) {$e_3: \mathit{Insert}(2, \text{``l''})$};
    \node (char4) at (0,0.7) {$e_4: \mathit{Insert}(3, \text{``o''})$};
    \node (char5) at (-1.5,0) {$e_5: \mathit{Insert}(3, \text{``l''})$};
    \node (char6) at (1.5,0) {$e_6: \mathit{Insert}(4, \text{``!''})$};
    \draw [->] (char1) -- (char2);
    \draw [->] (char2) -- (char3);
    \draw [->] (char3) -- (char4);
    \draw [->] (char4) -- (char5);
    \draw [->] (char4) -- (char6);
  \end{tikzpicture}
  \caption{The event graph corresponding to \autoref{two-inserts}.}
  \label{graph-example}
\end{figure}
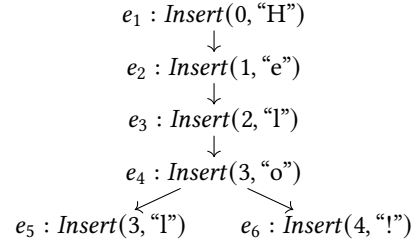

For example, \autoref{graph-example} shows the event graph corresponding to \autoref{two-inserts}.
The events $e_5$ and $e_6$ are concurrent, and the frontier of this graph is the set of events $\{e_5, e_6\}$.

The event graph for a substantial document, such as a research paper, may contain hundreds of thousands of events.
It can nevertheless be stored in a very compact form by exploiting the typical editing patterns of humans writing text: characters tend to be inserted or deleted in consecutive runs.
Many portions of a typical event graph are linear, with each event having one parent and one child.
We describe the storage format in more detail in \autoref{storage}.

\subsection{Document versions}\label{versions}

Let $G$ be an event graph, represented as a set of events.
Due to convergence, any two replicas that have the same set of events must be in the same state.
Therefore, the document state (sequence of characters) resulting from $G$ must be $\mathsf{replay}(G)$, where $\mathsf{replay}$ is some pure (deterministic and non-mutating) function.
In principle, any pure function of the set of events results in convergence, although a $\mathsf{replay}$ function that is useful for text editing must satisfy additional criteria (see \autoref{characteristics}).

Consider the event $\mathit{Delete}(i)$, which deletes the character at position $i$ in the document. In order to correctly interpret this event, we need to determine which character was at index $i$ at the time when the operation was generated.

More generally, let $e_i$ be some event. The document state when $e_i$ was generated must be $\mathsf{replay}(G_i)$, where $G_i$ is the set of events that were known to the generating replica at the time when $e_i$ was generated (not including $e_i$ itself).
By definition, the parents of $e_i$ are the frontier of $G_i$, and thus $G_i$ is the set of all events that happened before $e_i$, i.e., $e_i$'s parents and all of their ancestors.
Therefore, the parents of $e_i$ unambiguously define the document state in which $e_i$ must be interpreted.

To formalise this, given an event graph (set of events) $G$, we define the \emph{version} of $G$ to be its frontier set:
\begin{equation*}
  \mathsf{Version}(G) = \{e_1 \in G \mid \nexists e_2 \in G: e_1 \rightarrow e_2\}
\end{equation*}

Given some version $V$, the corresponding set of events can be reconstructed as follows:
\begin{equation*}
  \mathsf{Events}(V) = V \cup \{e_1 \mid \exists e_2 \in V : e_1 \rightarrow e_2\}
\end{equation*}

Since an event graph grows only by adding events that are concurrent to or children of existing events (we never change the parents of an existing event), there is a one-to-one correspondence between an event graph and its version.
For all valid event graphs $G$, $\mathsf{Events}(\mathsf{Version}(G)) = G$.

The set of parents of an event in the graph is the version of the document in which that operation must be interpreted.
The version can hence be seen as a \emph{logical clock}, describing the point in time at which a replica knows about the exact set of events in $G$.
Even if the event graph is large and there are many collaborators, a version rarely consists of more than two events in practice: a version with $n$ events occurs only if $n$ mutually concurrent events are merged with no new operations being generated in the intervening time.

\subsection{Replaying editing history}\label{replay}

Collaborative editing algorithms are usually defined in terms of sending and receiving messages over a network.
The abstraction of an event graph allows us to reframe these algorithms in a simpler way: a collaborative text editing algorithm is a pure function $\mathsf{replay}(G)$ of an event graph $G$.
This function can use the parent-child relationships to partially order events, but concurrent events could be processed in any order.
This allows us to separate the process of replicating the event graph from the algorithm that ensures convergence.
In fact, this is how \emph{pure operation-based CRDTs} \cite{polog} are formulated, as discussed in \autoref{related-work}.

In addition to determining the document state from an entire event graph, we need an \emph{incremental update} function.
Say we have an existing event graph $G$ and corresponding document state $\mathit{doc} = \mathsf{replay}(G)$. Then an event $e$ from a remote replica is added to the graph.
We could rerun the function to obtain $\mathit{doc}' = \mathsf{replay}(G \cup \{e\})$, but it would be inefficient to process the entire graph again.
Instead, we need to efficiently compute the operation to apply to $\mathit{doc}$ in order to obtain $\mathit{doc}'$.
For text documents, this incremental update is also described as an insertion or deletion at a particular index; however, the index may differ from that in the original event due to the effects of concurrent operations, and a deletion may turn into a no-op if the same character has also been deleted by a concurrent operation.

Both OT and CRDT algorithms focus on this incremental update.
If none of the events in $G$ are concurrent with $e$, OT is straightforward: the incremental update is identical to the operation in $e$, as no transformation takes place.
If there is concurrency, OT must transform each new event with regard to each existing event that is concurrent to it.

In CRDTs, each event is first translated into operations that use unique IDs instead of indexes, and then these operations are applied to a data structure that reflects all of the operations seen so far (both concurrent operations and those that happened before).
In order to update the text editor, these updates to the CRDT's internal structure need to be translated back into index-based insertions and deletions.
Many CRDT papers elide this translation from unique IDs back to indexes, but it is important for practical applications. 



Regardless of whether the OT or the CRDT approach is used, a collaborative editing algorithm can be boiled down to an incremental update to an event graph: given an event to be added to an existing event graph, return the (index-based) operation that must be applied to the current document state so that the resulting document is identical to replaying the entire event graph including the new event.


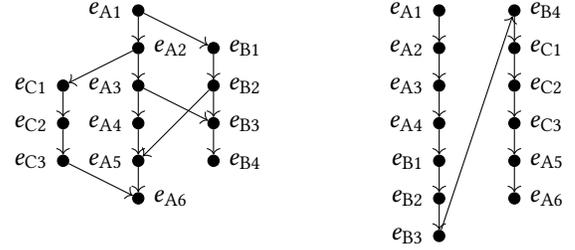
\begin{figure}
  \begin{tikzpicture}
    \tikzstyle{op} = [shape=circle,draw,fill,inner sep=1.5pt]
    \node [op] (a1) at (0,2) [label={left:$e_\mathrm{A1}$}] {};
    \node [op] (a2) at (0,1.5) [label={right:$e_\mathrm{A2}$}] {};
    \node [op] (a3) at (0,1) [label={left:$e_\mathrm{A3}$}] {};
    \node [op] (a4) at (0,0.5) [label={left:$e_\mathrm{A4}$}] {};
    \node [op] (a5) at (0,0) [label={left:$e_\mathrm{A5}$}] {};
    \node [op] (a6) at (0,-0.5) [label={right:$e_\mathrm{A6}$}] {};
    \node [op] (b1) at (1,1.5) [label={right:$e_\mathrm{B1}$}] {};
    \node [op] (b2) at (1,1) [label={right:$e_\mathrm{B2}$}] {};
    \node [op] (b3) at (1,0.5) [label={right:$e_\mathrm{B3}$}] {};
    \node [op] (b4) at (1,0) [label={right:$e_\mathrm{B4}$}] {};
    \node [op] (c1) at (-1,1) [label={left:$e_\mathrm{C1}$}] {};
    \node [op] (c2) at (-1,0.5) [label={left:$e_\mathrm{C2}$}] {};
    \node [op] (c3) at (-1,0) [label={left:$e_\mathrm{C3}$}] {};
    \node [op] (x1) at (4,2) [label={left:$e_\mathrm{A1}$}] {};
    \node [op] (x2) at (4,1.5) [label={left:$e_\mathrm{A2}$}] {};
    \node [op] (x3) at (4,1) [label={left:$e_\mathrm{A3}$}] {};
    \node [op] (x4) at (4,0.5) [label={left:$e_\mathrm{A4}$}] {};
    \node [op] (x5) at (4,0) [label={left:$e_\mathrm{B1}$}] {};
    \node [op] (x6) at (4,-0.5) [label={left:$e_\mathrm{B2}$}] {};
    \node [op] (x7) at (4,-1) [label={left:$e_\mathrm{B3}$}] {};
    \node [op] (x8) at (5,2) [label={right:$e_\mathrm{B4}$}] {};
    \node [op] (x9) at (5,1.5) [label={right:$e_\mathrm{C1}$}] {};
    \node [op] (x10) at (5,1) [label={right:$e_\mathrm{C2}$}] {};
    \node [op] (x11) at (5,0.5) [label={right:$e_\mathrm{C3}$}] {};
    \node [op] (x12) at (5,0) [label={right:$e_\mathrm{A5}$}] {};
    \node [op] (x13) at (5,-0.5) [label={right:$e_\mathrm{A6}$}] {};
    \draw [->] (a1) -- (a2);
    \draw [->] (a2) -- (a3);
    \draw [->] (a3) -- (a4);
    \draw [->] (a4) -- (a5);
    \draw [->] (a5) -- (a6);
    \draw [->] (b1) -- (b2);
    \draw [->] (b2) -- (b3);
    \draw [->] (b3) -- (b4);
    \draw [->] (c1) -- (c2);
    \draw [->] (c2) -- (c3);
    \draw [->] (a1) -- (b1);
    \draw [->] (a3) -- (b3);
    \draw [->] (b2) -- (a5);
    \draw [->] (a2) -- (c1);
    \draw [->] (c3) -- (a6);
    \draw [->] (x1) -- (x2);
    \draw [->] (x2) -- (x3);
    \draw [->] (x3) -- (x4);
    \draw [->] (x4) -- (x5);
    \draw [->] (x5) -- (x6);
    \draw [->] (x6) -- (x7);
    \draw [->] (x7) -- (x8);
    \draw [->] (x8) -- (x9);
    \draw [->] (x9) -- (x10);
    \draw [->] (x10) -- (x11);
    \draw [->] (x11) -- (x12);
    \draw [->] (x12) -- (x13);
  \end{tikzpicture}
  \caption{An event graph (left) and one possible topologically sorted order of that graph (right).}
  \label{topological-sort}
\end{figure}

\subsection{Implementing OT using a CRDT}\label{crdt-replay}

One way of implementing such a replay algorithm would be to simulate a network of CRDT replicas in a single process.
For each branch in the event graph there is a separate simulated replica, which takes operations in their original index-based form and generates a corresponding ID-based CRDT operation.
Another simulated replica receives every operation generated by the other replicas and applies them in some topologically sorted order, as illustrated in \autoref{topological-sort}.

For example, the history in \autoref{topological-sort} could be replayed using one simulated replica for $e_\mathrm{A1\dots A6}$, a second for $e_\mathrm{B1\dots B4}$, and a third for $e_\mathrm{C1\dots C3}$.
Every time an event's parent is an event generated on another simulated replica, the corresponding network communication is simulated, and the remote operations are merged using a CRDT algorithm.
For example, before the replica for $e_\mathrm{B1\dots B4}$ can generate $e_\mathrm{B3}$ it must first merge $e_\mathrm{A2}$ and $e_\mathrm{A3}$.
Each simulated replica thus tracks the document version in which the indexes of insertions and deletions should be interpreted.
The simulated replica that applies all operations then converts the ID-based operation back into an index based on its document version.
This index-based operation then allows an incremental update of the document state.

The process of translating an index-based operation into an ID-based one on one simulated replica, and translating it back into an index-based operation on another, is effectively an operational transformation algorithm: it updates the index to reflect the effects of concurrent operations (which have been applied to the second simulated replica but not the first).
However, the algorithm is fairly slow because it incurs the overhead of updating multiple simulated replicas and running the CRDT algorithm even at times when there is no concurrency in the event graph.
It also uses a lot of memory because it needs a separate copy of the CRDT state for every concurrent branch in the event graph.

Our \algname algorithm, described in the next section, modifies this approach to use only two simulated replicas: one on which operations are generated, and the other on which all operations are applied (and in fact, both are stored in the same data structure).
To deal with event graphs that are not totally ordered, the algorithm allows events on one branch to be \emph{retreated} when switching to another branch, and \emph{advanced} again when those branches are merged.
Retreating an event updates the replica state to behave as if that event had not yet happened, and advancing makes the event take effect again.

For example, in \autoref{topological-sort}, after applying $e_\mathrm{A4}$ we would retreat $e_\mathrm{A4}$, $e_\mathrm{A3}$, and $e_\mathrm{A2}$ before applying $e_\mathrm{B1}$, since those events are concurrent with $e_\mathrm{B1}$.
Before applying $e_\mathrm{B3}$ we would advance $e_\mathrm{A2}$ and $e_\mathrm{A3}$ again, since they are ancestors of $e_\mathrm{B3}$.
Retreating and advancing takes some additional CPU time on highly concurrent event graphs, but as we show in \autoref{benchmarking}, the optimisations this approach enables result in excellent performance overall.

\section{The Event Graph Walker algorithm}\label{algorithm}

\algname is a collaborative text editing algorithm based on the idea of event graph replay.
The algorithm builds on a replication layer that ensures that whenever a replica adds an event to the graph, all non-crashed replicas eventually receive it.
The state of each replica consists of three parts:

\begin{enumerate}
\item \textbf{Event graph:} Each replica stores a copy of the event graph on disk, in a format described in \autoref{storage}.
\item \textbf{Document state:} The current sequence of characters in the document with no further metadata. On disk this is simply a plain text file; in memory it may be represented as a rope \cite{Boehm1995}, piece table \cite{vscode-buffer}, or similar structure to support efficient insertions and deletions.
\item \textbf{Internal state:} A temporary CRDT structure that \algname uses to merge concurrent edits. It is not persisted or replicated, and it is discarded when the algorithm finishes running.
\end{enumerate}

\algname can reconstruct the document state by replaying the entire event graph.
It first performs a topological sort, as illustrated in \autoref{topological-sort}. Then each event is transformed so that the transformed insertions and deletions can be applied in topologically sorted order, starting with an empty document, to obtain the document state.
In Git parlance, this process ``rebases'' a DAG of operations into a linear operation history with the same effect.
The input of the algorithm is the event graph, and the output is this topologically sorted sequence of transformed operations.
While OT transforms one operation with respect to one other, \algname uses the internal state to transform sets of operations efficiently.

In graphs with concurrent operations there are multiple possible sort orders. \algname guarantees that the final document state is the same, regardless which of these orders is chosen. However, the choice of sort order may affect the performance of the algorithm, as discussed in \autoref{complexity}.

For example, the graph in \autoref{graph-example} has two possible sort orders; \algname either first inserts ``l'' at index 3 and then ``!'' at index 5 (like User 1 in \autoref{two-inserts}), or it first inserts ``!'' at index 4 followed by ``l'' at index 3 (like User 2 in \autoref{two-inserts}). The final document state is ``Hello!'' either way.

Event graph replay easily extends to incremental updates for real-time collaboration: when a new event is added to the graph, it becomes the next element of the topologically sorted sequence.
We can transform each new event in the same way as during replay, and apply the transformed operation to the current document state.

\subsection{Characteristics of \algname}\label{characteristics}

\algname ensures that the resulting document state is consistent with Attiya et al.'s \emph{strong list specification} \cite{Attiya2016} (in essence, replicas converge to the same state and apply operations in the right place), and it is \emph{maximally non-interleaving} \cite{fugue} (i.e., concurrent sequences of insertions at the same position are placed one after another, and not interleaved).

When generating new events, or when adding an event to the graph that happened after all existing events, \algname only needs the current document state.
Most of the time, the event graph can thus remain on disk without using any space in memory or any CPU time, and the internal state can be discarded entirely.
The event graph and internal state are only required when handling concurrency, and even then we only have to replay the portion of the graph since the last ancestor that the concurrent operations had in common.
In portions of the event graph that have no concurrency (which, in many editing histories, is the vast majority of events), events do not need to be transformed at all.

In contrast, existing CRDTs require every replica to persist the internal state and send it over the network.
They also require that state to be loaded into memory to generate and receive operations, even when there is no concurrency.
This uses several times more memory and makes documents slow to load.

OT algorithms avoid this internal state; similarly to \algname, they only need to persist the latest document state and the history of operations that are concurrent to operations that may arrive in the future.
In both \algname and OT, the event graph can be discarded if we know that no event we may receive in the future will be concurrent with any existing event.
However, OT algorithms are very slow to merge long-running branches (see \autoref{benchmarking}).
\algname handles arbitrary event DAGs, whereas some OT algorithms are only able to handle restricted forms of event graphs (server-based OT corresponds to event graphs with one main branch representing the server's view; all other branches may merge to and from the main branch, but not with each other).

\subsection{Walking the event graph}\label{graph-walk}

For the sake of clarity we first explain a simplified version of \algname that replays the entire event graph without discarding its internal state along the way. This approach incurs some CRDT overhead even for non-concurrent operations.
We give pseudocode for this simplified algorithm in \ifincludeappendix\autoref{pseudocode-appendix}\else the extended version of this paper \cite{extended-version}\fi.
In \autoref{partial-replay} we show how the algorithm can be optimised to replay only a part of the event graph.

First, we topologically sort the event graph in a way that keeps events on the same branch consecutive as much as possible: for example, in \autoref{topological-sort} we first visit $e_\mathrm{A1} \dots e_\mathrm{A4}$, then $e_\mathrm{B1} \dots e_\mathrm{B4}$. We avoid alternating between branches, such as $e_\mathrm{A1}, e_\mathrm{B1}, e_\mathrm{A2}, e_\mathrm{B2} \dots$, even though that would also be a valid topological sort.
For this we use a standard textbook algorithm \cite{CLRS2009}: perform a depth-first traversal starting from the oldest event, and build up the topologically sorted list in the order that events are visited.
When a node has multiple children in the graph, we choose their order based on a heuristic so that branches with fewer events tend to appear before branches with more events in the sorted order; this can improve performance (see \autoref{complexity}) but is not essential.
We estimate the size of a branch by counting the number of events that happened after each event.

The algorithm then processes the events one at a time in topologically sorted order, updating the internal state and outputting a transformed operation for each event.
The internal state simultaneously captures the document at two versions: the version in which an event was generated (which we call the \emph{prepare} version), and the version in which all events seen so far have been applied (which we call the \emph{effect} version).
These correspond to the two simulated replicas mentioned in \autoref{crdt-replay}.
If the prepare and effect versions are the same, the transformed operation is identical to the original one.
In general, the prepare version represents a subset of the events of the effect version.

The internal state can be updated with three methods, each of which takes an event as argument:

\begin{itemize}
\item $\mathsf{apply}(e)$ updates the prepare version and the effect version to include $e$, assuming that the current prepare version equals $e.\mathit{parents}$, and that $e$ has not yet been applied. This method interprets $e$ in the context of the prepare version, and outputs the operation representing how the effect version has been updated.
\item $\mathsf{retreat}(e)$ updates the prepare version to remove $e$, assuming the prepare version previously included $e$.
\item $\mathsf{advance}(e)$ updates the prepare version to add $e$, assuming that the prepare version previously did not include $e$, but the effect version did.
\end{itemize}

\begin{figure}
  \begin{tikzpicture}[every node/.style={inner ysep=1pt}]
    \node (e1) at (2,3.5) {$e_1: \mathit{Insert}(0, \text{``h''})$};
    \node (e2) at (2,2.8) {$e_2: \mathit{Insert}(1, \text{``i''})$};
    \node (e3) at (0,2.1) {$e_3: \mathit{Insert}(0, \text{``H''})$};
    \node (e4) at (0,1.4) {$e_4: \mathit{Delete}(1)$};
    \node (e5) at (4,2.1) {$e_5: \mathit{Delete}(1)$};
    \node (e6) at (4,1.4) {$e_6: \mathit{Insert}(1, \text{``e''})$};
    \node (e7) at (4,0.7) {$e_7: \mathit{Insert}(2, \text{``y''})$};
    \node (e8) at (2,0.0) {$e_8: \mathit{Insert}(3, \text{``!''})$};
    \draw [->] (e1) -- (e2);
    \draw [->] (e2) -- (e3);
    \draw [->] (e3) -- (e4);
    \draw [->] (e2) -- (e5);
    \draw [->] (e5) -- (e6);
    \draw [->] (e6) -- (e7);
    \draw [->] (e7) -- (e8);
    \draw [->] (e4) -- (e8);
  \end{tikzpicture}
  \caption{An event graph. Starting with document ``hi'', one user changes ``hi'' to ``hey'', while concurrently another user capitalises the ``H''. After merging to the state ``Hey'', one of them appends an exclamation mark to produce ``Hey!''.}
  \label{graph-hi-hey}
\end{figure}
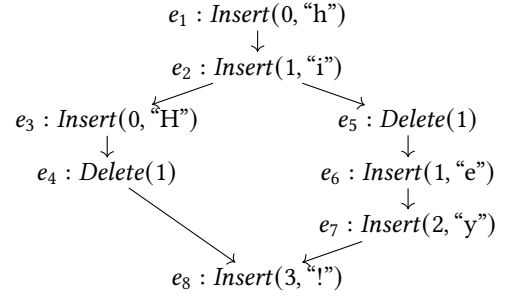

The effect version only moves forwards in time (through $\mathsf{apply}$), whereas the prepare version can move both forwards and backwards.
Consider the example in \autoref{graph-hi-hey}, and assume that the events $e_1 \dots e_8$ are traversed in order of their subscript.
These events can be processed as follows:

\begin{enumerate}
\item Start in the empty state, then call $\mathsf{apply}(e_1)$, $\mathsf{apply}(e_2)$, $\mathsf{apply}(e_3)$, and $\mathsf{apply}(e_4)$. This is valid because each event's parent version is the previously applied event.
\item Before we can apply $e_5$ we must rewind the prepare version to be $\{e_2\}$, which is the parent of $e_5$. We can do this by calling $\mathsf{retreat}(e_4)$ and $\mathsf{retreat}(e_3)$.
\item Now we can call $\mathsf{apply}(e_5)$, $\mathsf{apply}(e_6)$, and $\mathsf{apply}(e_7)$.
\item The parents of $e_8$ are $\{e_4, e_7\}$; before we can apply $e_8$ we must therefore add $e_3$ and $e_4$ to the prepare state again by calling $\mathsf{advance}(e_3)$ and $\mathsf{advance}(e_4)$.
    We do not retreat $e_{5\dots 7}$ because $e_3$ and $e_4$ have already been applied in Step 1; now we are \emph{advancing} $e_3$ and $e_4$, which does not require retreating concurrent events.
\item Finally, we can call $\mathsf{apply}(e_8)$.
\end{enumerate}

In complex event graphs such as the one in \autoref{topological-sort} the same event may have to be retreated and advanced several times, but we can process arbitrary DAGs this way.
In general, before applying the next event $e$ in topologically sorted order, compute $G_\mathrm{old} = \mathsf{Events}(V_p)$ where $V_p$ is the current prepare version, and $G_\mathrm{new} = \mathsf{Events}(e.\mathit{parents})$.
We then call $\mathsf{retreat}$ on each event in $G_\mathrm{old} - G_\mathrm{new}$ (in reverse topological sort order), and call $\mathsf{advance}$ on each event in $G_\mathrm{new} - G_\mathrm{old}$ (in topological sort order) before calling $\mathsf{apply}(e)$.

The following algorithm efficiently computes the events to retreat and advance when moving the prepare version from $V_p$ to $V'_p$.
For each event in $V_p$ and $V'_p$ we insert the index of that event in the topological sort order into a priority queue, along with a tag indicating whether the event is in the old or the new prepare version.
We then repeatedly pop the event with the greatest index off the priority queue, and enqueue the indexes of its parents along with the same tag.
We stop the traversal when all entries in the priority queue are common ancestors of both $V_p$ and $V'_p$.
Any events that were traversed from only one of the versions need to be retreated or advanced respectively.

\subsection{Representing prepare and effect versions}\label{prepare-effect-versions}

The internal state implements the $\mathsf{apply}$, $\mathsf{retreat}$, and $\mathsf{advance}$ methods by maintaining a CRDT data structure.
This structure consists of a linear sequence of records, one per character in the document, including tombstones for deleted characters.
Runs of characters with consecutive IDs and the same properties can be run-length encoded to save memory.
A record is inserted into this sequence by $\mathsf{apply}(e_i)$ for an insertion event $e_i$.
Subsequent deletion events and $\mathsf{retreat}$/$\mathsf{advance}$ calls may modify properties of the record, but records in the sequence are not removed or reordered once they have been inserted.

When the event graph contains concurrent insertions, we use a CRDT to ensure that all replicas place the records in this sequence in the same order, regardless of the order in which the event graph is traversed.
For example, RGA \cite{Roh2011RGA} or YATA \cite{Nicolaescu2016YATA} could be used for this purpose.
Our implementation of \algname uses a variant of the Yjs algorithm \cite{yjs}, itself based on YATA, that we conjecture to be maximally non-interleaving.
We leave a detailed analysis of this algorithm to future work, since it is not core to this paper.

Each record in this sequence contains:
\begin{itemize}
\item the ID of the event that inserted the character;
\item $s_p \in \{\texttt{NotInsertedYet}, \texttt{Ins}, \texttt{Del 1}, \texttt{Del 2}, \dots\}$, the character's state in the prepare version;
\item $s_e \in \{\texttt{Ins}, \texttt{Del}\}$, the state in the effect version;
\item and any other fields required by the CRDT to determine the order of concurrent insertions.
\end{itemize}

The rules for updating $s_p$ and $s_e$ are:

\begin{itemize}
\item When a record is first inserted by $\mathsf{apply}(e_i)$ with an insertion event $e_i$, it is initialised with $s_p = s_e = \texttt{Ins}$.
\item If $\mathsf{apply}(e_d)$ is called with a deletion event $e_d$, we set $s_e = \texttt{Del}$ in the record representing the deleted character. In the same record, if $s_p = \texttt{Ins}$ we update it to $\texttt{Del 1}$, and if $s_p = \texttt{Del}\; n$ it advances to $\texttt{Del} (n+1)$, as shown in \autoref{spv-state}.
\item If $\mathsf{retreat}(e_i)$ is called with insertion event $e_i$, we must have $s_p = \texttt{Ins}$ in the record affected by the event, and we update it to $s_p = \texttt{NotInsertedYet}$. Conversely, $\mathsf{advance}(e_i)$ moves $s_p$ from $\texttt{NotInsertedYet}$ to $\texttt{Ins}$.
\item If $\mathsf{retreat}(e_d)$ is called with a deletion event $e_d$, we must have $s_p = \texttt{Del}\; n$ in the affected record, and we update it to $\texttt{Del} (n-1)$ if $n>1$, or to $\texttt{Ins}$ if $n=1$. Calling $\mathsf{advance}(e_d)$ performs the opposite.
\end{itemize}

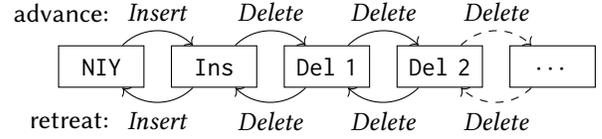
\begin{figure}
  \begin{tikzpicture}
    \tikzstyle{state} = [draw,align=center,text width=0.9cm,text height=8pt,text depth=1pt]
    \node [state] (nyi) at (0, 0) {\texttt{NIY}};
    \node [state] (ins) at (1.5, 0) {\texttt{Ins}};
    \node [state] (del1) at (3, 0) {\texttt{Del 1}};
    \node [state] (del2) at (4.5, 0) {\texttt{Del 2}};
    \node [state] (deln) at (6, 0) {$\cdots$};
    \draw [->] (nyi) to [out=45,in=135] node [above,label={left:$\mathsf{advance}$:}] {$\mathit{Insert}$} (ins);
    \draw [->] (ins) to [out=45,in=135] node [above] {$\mathit{Delete}$} (del1);
    \draw [->] (del1) to [out=45,in=135] node [above] {$\mathit{Delete}$} (del2);
    \draw [->,dashed] (del2) to [out=45,in=135] node [above] {$\mathit{Delete}$} (deln);
    \draw [->] (ins) to [out=225,in=315] node [below,label={left:$\mathsf{retreat}$:}] {$\mathit{Insert}$} (nyi);
    \draw [->] (del1) to [out=225,in=315] node [below] {$\mathit{Delete}$} (ins);
    \draw [->] (del2) to [out=225,in=315] node [below] {$\mathit{Delete}$} (del1);
    \draw [->,dashed] (deln) to [out=225,in=315] node [below] {$\mathit{Delete}$} (del2);
  \end{tikzpicture}
  \caption{State machine for internal state variable $s_p$.}
  \label{spv-state}
\end{figure}

As a result, $s_p$ and $s_e$ are \texttt{Ins} if the character is visible (inserted but not deleted) in the prepare and effect version respectively; $s_p = \texttt{Del}\; n$ indicates that the character has been deleted by $n$ concurrent delete events in the prepare version; and $s_p = \texttt{NotInsertedYet}$ indicates that the insertion of the character has been retreated in the prepare version.
$s_e$ does not count the number of deletions and does not have a $\texttt{NotInsertedYet}$ state since we never remove the effect of an operation from the effect version.

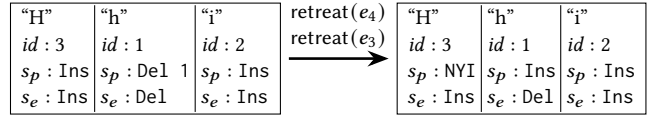
\begin{figure}
  \begin{tikzpicture}[font=\footnotesize]
    \tikzstyle{box}=[rectangle,draw,matrix,text height=6pt,text depth=2pt,nodes=right,inner xsep=2pt,inner ysep=1pt,column sep=0]
    \matrix [box] (left) at (0,0) {
      \node {``H''}; & \node (top1) {``h''}; & \node (top2) {``i''}; \\
      \node {$\mathit{id}: 3$}; & \node {$\mathit{id}: 1$}; & \node {$\mathit{id}: 2$}; \\
      \node {$s_p: \texttt{Ins}$}; & \node {$s_p: \texttt{Del 1}$}; & \node {$s_p: \texttt{Ins}$}; \\
      \node {$s_e: \texttt{Ins}$}; & \node (bot1) {$s_e: \texttt{Del}$}; & \node (bot2) {$s_e: \texttt{Ins}$}; \\
    };
    \matrix [box] (right) at (5,0) {
      \node {``H''}; & \node (top3) {``h''}; & \node (top4) {``i''}; \\
      \node {$\mathit{id}: 3$}; & \node {$\mathit{id}: 1$}; & \node {$\mathit{id}: 2$}; \\
      \node {$s_p: \texttt{NYI}$}; & \node {$s_p: \texttt{Ins}$}; & \node {$s_p: \texttt{Ins}$}; \\
      \node {$s_e: \texttt{Ins}$}; & \node (bot3) {$s_e: \texttt{Del}$}; & \node (bot4) {$s_e: \texttt{Ins}$}; \\
    };
    \draw (top1.north west) -- (bot1.south west);
    \draw (top2.north west) -- (bot2.south west);
    \draw (top3.north west) -- (bot3.south west);
    \draw (top4.north west) -- (bot4.south west);
    \draw [-{Stealth[length=8pt]},line width=1pt] ([xshift=3pt]left.east) -- node [above,align=center] {$\mathsf{retreat}(e_4)$\\$\mathsf{retreat}(e_3)$} ([xshift=-3pt]right.west);
  \end{tikzpicture}
  \caption{Left: the internal state after applying $e_1 ... e_4$ from \autoref{graph-hi-hey}. Right: after $\mathsf{retreat}(e_4)$ and $\mathsf{retreat}(e_3)$, the prepare state is updated to mark ``H'' as \texttt{NotInsertedYet}, and the deletion of ``h'' is undone. The effect state is unchanged.}
  \label{crdt-state-1}
\end{figure}

For example, \autoref{crdt-state-1} shows the internal state after applying $e_1 \dots e_4$ from \autoref{graph-hi-hey}, and how that state is updated by retreating $e_4$ and $e_3$ before $e_5$ is applied.
In the effect state, the lowercase ``h'' is marked as deleted, while the uppercase ``H'' and the ``i'' are visible.
In the prepare state, by retreating $e_4$ and $e_3$ the ``H'' is marked as \texttt{NotInsertedYet}, and the deletion of ``h'' is undone ($s_p = \texttt{Ins}$).

\begin{figure}
  \begin{tikzpicture}[font=\footnotesize]
    \tikzstyle{box}=[rectangle,draw,matrix,text height=6pt,text depth=2pt,nodes=right,inner xsep=2pt,inner ysep=1pt,column sep=0]
    \matrix [box] at (0,0) {
      \node {``H''}; & \node (top1) {``h''}; & \node (top2) {``e''}; & \node (top3) {``y''}; & \node (top4) {``!''}; & \node (top5) {``i''}; \\
      \node {$\mathit{id}: 3$}; & \node {$\mathit{id}: 1$}; & \node {$\mathit{id}: 6$}; & \node {$\mathit{id}: 7$}; & \node {$\mathit{id}: 8$}; & \node {$\mathit{id}: 2$}; \\
      \node {$s_p: \texttt{Ins}$}; & \node {$s_p: \texttt{Del 1}$}; & \node {$s_p: \texttt{Ins}$}; & \node {$s_p: \texttt{Ins}$}; & \node {$s_p: \texttt{Ins}$}; & \node {$s_p: \texttt{Del 1}$}; \\
      \node {$s_e: \texttt{Ins}$}; & \node (bot1) {$s_e: \texttt{Del}$}; & \node (bot2) {$s_e: \texttt{Ins}$}; & \node (bot3) {$s_e: \texttt{Ins}$}; & \node (bot4) {$s_e: \texttt{Ins}$}; & \node (bot5) {$s_e: \texttt{Del}$}; \\
    };
    \draw (top1.north west) -- (bot1.south west);
    \draw (top2.north west) -- (bot2.south west);
    \draw (top3.north west) -- (bot3.south west);
    \draw (top4.north west) -- (bot4.south west);
    \draw (top5.north west) -- (bot5.south west);
  \end{tikzpicture}
  \caption{The internal \algname state after replaying all of the events in \autoref{graph-hi-hey}.}
  \label{crdt-state-2}
\end{figure}

\autoref{crdt-state-2} shows the state after replaying all of the events in \autoref{graph-hi-hey}: ``i'' is also deleted, the characters ``e'' and ``y'' are inserted immediately after the ``h'', $e_3$ and $e_4$ are advanced again, and finally ``!'' is inserted after the ``y''.
The figures include the character for the sake of readability, but \algname actually does not store text content in its internal state.

\subsection{Mapping indexes to character IDs}\label{b-trees}

In the event graph, insertion and deletion operations specify the index at which they apply.
In order to update \algname's internal state, we need to map these indexes to the correct record in the sequence, based on the prepare state $s_p$.
To produce the transformed operations, we need to map the positions of these internal records back to indexes again -- this time based on the effect state $s_e$.

A simple but inefficient algorithm would be: to apply a $\mathit{Delete}(i)$ operation we iterate over the sequence of records and pick the $i$th record with a prepare state of $s_p = \texttt{Ins}$ (i.e., the $i$th among the characters that are visible in the prepare state, which is the document state in which the operation should be interpreted).
Similarly, to apply $\mathit{Insert}(i, c)$ we skip over $i - 1$ records with $s_p = \texttt{Ins}$ and insert the new record after the last skipped record (if there have been concurrent insertions at the same position, we may also need to skip over some records with $s_p = \texttt{NotInsertedYet}$, as determined by the list CRDT's insertion ordering).

To reduce the cost of this algorithm from $O(n)$ to $O(\log n)$, where $n$ is the number of characters in the document, we construct a B-tree whose leaves, from left to right, contain the sequence of records representing characters.
We extend the tree into an \emph{order statistic tree} \cite{CLRS2009} (also known as \emph{ranked B-tree}) by adding two integers to each node: the number of records with $s_p = \texttt{Ins}$ contained within that subtree, and the number of records with $s_e = \texttt{Ins}$ in that subtree.
Every time $s_p$ or $s_e$ are updated, we also update those numbers on the path from the updated record to the root.
As the tree is balanced, this update takes $O(\log n)$.

Now we can find the $i$th record with $s_p = \texttt{Ins}$ in logarithmic time by starting at the root of the tree, and adding up the values in the subtrees that have been skipped.
Moreover, once we have a record in the sequence we can efficiently determine its index in the effect state by going in the opposite direction: working upwards in the tree towards the root, and summing the numbers of records with $s_e = \texttt{Ins}$ that lie in subtrees to the left of the starting record.
This allows us to efficiently transform the index of an operation from the prepare version into the effect version.
If the character was already deleted in the effect version ($s_e = \texttt{Del}$), the transformed operation is a no-op.

The above process makes $\mathsf{apply}(e_i)$ efficient.
We also need to efficiently perform $\mathsf{retreat}(e_i)$ and $\mathsf{advance}(e_i)$, which modify the prepare state $s_p$ of the record inserted or deleted by $e_i$.
While advancing/retreating we cannot look up a target record by its index. Instead, we maintain a second B-tree, mapping from each event's ID to the target record. The mapping stores a value depending on the type of the event:

\begin{itemize}
\item For delete events, we store the ID of the character deleted by the event.
\item For insert events, we store a pointer to the leaf node in the first B-tree that contains the corresponding record. When nodes in the first B-tree are split, we update the pointers in the second B-tree accordingly.
\end{itemize}

On every $\mathsf{apply}(e)$, after updating the sequence as above, we update this mapping.
When we later call $\mathsf{retreat}(e)$ or $\mathsf{advance}(e)$, that event $e$ must have already been applied, and hence $e.\mathit{id}$ must appear in this mapping.
This map allows us to advance or retreat in logarithmic time.

\subsection{Clearing the internal state}\label{clearing}

As described so far, the algorithm retains every insertion since document creation forever in its internal state, consuming a lot of memory, and requiring the entire event graph to be replayed in order to restore the internal state.
We now introduce a further optimisation that allows \algname to completely discard its internal state from time to time, and replay only a subset of the event graph.

We define a version $V \subseteq G$ to be a \emph{critical version} in an event graph $G$ iff it partitions the graph into two subsets of events $G_1 = \mathsf{Events}(V)$ and $G_2 = G - G_1$ such that all events in $G_1$ happened before all events in $G_2$:
\begin{equation*}
  \forall e_1 \in G_1: \forall e_2 \in G_2: e_1 \rightarrow e_2.
\end{equation*}

Equivalently, $V$ is a critical version iff every event in the graph is either in $V$, or an ancestor of some event in $V$, or happened after \emph{all} of the events in $V$:
\begin{equation*}
  \forall e_1 \in G: e_1 \in \mathsf{Events}(V) \vee (\forall e_2 \in V: e_2 \rightarrow e_1).
\end{equation*}
A critical version might not remain critical forever; it is possible for a critical version to become non-critical because a concurrent event is added to the graph.

A key insight in the design of \algname is that critical versions partition the event graph into sections that can be processed independently. Events that happened at or before a critical version do not affect how any event after the critical version is transformed. 
This observation enables two important optimisations:

\begin{itemize}
\item Any time the version of the event graph processed so far is critical, we can discard the internal state (including both B-trees and all $s_p$ and $s_e$ values), and replace it with a placeholder as explained in \autoref{partial-replay}.
\item If both an event's version and its parent version are critical versions, there is no need to traverse the B-trees and update the CRDT state, since we would immediately discard that state anyway. In this case, the transformed event is identical to the original event, so the event can simply be emitted as-is.
\end{itemize}

These optimisations make it very fast to process documents that are mostly edited sequentially (e.g., because the authors took turns and did not write concurrently, or because there is only a single author), since most of the event graph of such a document is a linear chain of critical versions.

The internal state can be discarded once replay is complete, although it is also possible to retain the internal state for transforming future events.
If a replica receives events that are concurrent with existing events in its graph, but the replica has already discarded its internal state resulting from those events, it needs to rebuild some of that state.
It can do this by identifying the most recent critical version that happened before the new events, replaying the existing events that happened after that critical version, and finally applying the new events.
Events from before that critical version are not replayed.
Since most editing histories have critical versions from time to time, this means that usually only a small subset of the event graph is replayed.
In the worst case, this algorithm replays the entire event graph.

\subsection{Partial event graph replay}\label{partial-replay}

Assume that we want to add event $e_\mathrm{new}$ to the event graph $G$, that $V_\mathrm{curr} = \mathsf{Version}(G)$ is the current document version reflecting all events except $e_\mathrm{new}$, and that $V_\mathrm{crit} \neq V_\mathrm{curr}$ is the latest critical version in $G \cup \{e_\mathrm{new}\}$ that happened before both $e_\mathrm{new}$ and $V_\mathrm{curr}$.
Further assume that we have discarded the internal state, so the only information we have is the latest document state at $V_\mathrm{curr}$ and the event graph; in particular, without replaying the entire event graph we do not know the document state at $V_\mathrm{crit}$.

Luckily, the exact internal state at $V_\mathrm{crit}$ is not needed. All we need is enough state to transform $e_\mathrm{new}$ and rebase it onto the document at $V_\mathrm{curr}$.
This internal state can be obtained by replaying the events since $V_\mathrm{crit}$, that is, $G - \mathsf{Events}(V_\mathrm{crit})$, in topologically sorted order:

\begin{enumerate}
\item We initialise a new internal state corresponding to version $V_\mathrm{crit}$. Since we do not know the the document state at this version, we start with a single placeholder record representing the unknown document content.
\item We update the internal state by replaying events from $V_\mathrm{crit}$ to $V_\mathrm{curr}$, but we do not output transformed operations during this stage.
\item Finally, we apply the new event $e_\mathrm{new}$ and output the transformed operation. If we received a batch of new events, we apply them in topologically sorted order.
\end{enumerate}

The placeholder record we start with in step 1 represents the range of indexes $[0, \infty]$ of the document state at $V_\mathrm{crit}$ (we do not know the length of the document at that version, but we can still have a placeholder for arbitrarily many indexes).
Placeholders are counted as the number of characters they represent in the order statistic tree construction, and they have the same length in both the prepare and the effect versions.
We then apply events as follows:

\begin{itemize}
\item Applying an insertion at index $i$ creates a record with $s_p = s_e = \texttt{Ins}$ and the ID of the insertion event. We map the index to a record in the sequence using the prepare state as usual; if $i$ falls within a placeholder for range $[j, k]$, we split it into a placeholder for $[j, i-1]$, followed by the new record, followed by a placeholder for $[i, k]$. Placeholders for empty ranges are omitted.
\item Applying a deletion at index $i$: if the deleted character was inserted prior to $V_\mathrm{crit}$, the index must fall within a placeholder with some range $[j, k]$. We split it into a placeholder for $[j, i-1]$, followed by a new record with $s_p = \texttt{Del 1}$ and $s_e = \texttt{Del}$, followed by a placeholder for $[i+1, k]$. The new record has a placeholder ID that only needs to be unique within the local replica, and need not be consistent across replicas.
\item Applying a deletion of a character inserted since $V_\mathrm{crit}$ updates the record created by the insertion.
\end{itemize}

Before applying an event we retreat and advance as usual.
The algorithm never needs to retreat or advance an event that happened before $V_\mathrm{crit}$, therefore every retreated or advanced event ID must exist in second B-tree.

If there are concurrent insertions at the same position, we invoke the CRDT algorithm to place them in a consistent order as discussed in \autoref{prepare-effect-versions}.
Since all concurrent events must be after $V_\mathrm{crit}$, they are included in the replay.
When we are seeking for the insertion position, we never need to seek past a placeholder, since the placeholder represents characters that were inserted before $V_\mathrm{crit}$.

\subsection{Algorithm complexity}\label{complexity}

Say we have two users who have been working offline, generating $k$ and $m$ events respectively.
When they come online and merge their event graphs, the latest critical version is immediately prior to the branching point.
If the branch of $k$ events comes first in the topological sort, the replay algorithm first applies $k$ events, then retreats $k$ events, applies $m$ events, and finally advances $k$ events again.
Asymptotically, $O(k+m)$ calls to apply/retreat/advance are required regardless of the order of traversal, although in practice the algorithm is faster if $k<m$ since we don't need to retreat/advance on the branch that is visited last.

Each apply/retreat/advance requires one or two traversals of first B-tree, and at most one traversal of the second B-tree.
The upper bound on the number of entries in each tree (including placeholders) is $2(k+m)+1$, since each event generates at most one new record and one placeholder split.
Since the trees are balanced, the cost of each traversal is $O(\log(k+m))$.
Overall, the cost of merging branches with $k$ and $m$ events is therefore $O((k+m) \log(k+m))$.

We can also give an upper bound on the complexity of replaying an arbitrary event graph with $n$ events.
Each event is applied exactly once, and before each event we retreat or advance each prior event at most once, at $O(\log n)$ cost.
The worst-case complexity of the algorithm is therefore $O(n^2 \log n)$, but this case is unlikely to occur in practice.

\subsection{Storing the event graph}\label{storage}

To store the event graph compactly on disk, we developed a compression technique that takes advantage of how people typically write text documents: namely, they tend to insert or delete consecutive sequences of characters, and less frequently hit backspace or move the cursor to a new location.
\algname's event graph storage format is inspired by the Automerge CRDT library \cite{automerge-storage,automerge-columnar}, which in turn uses ideas from column-oriented databases \cite{Abadi2013,Stonebraker2005}. We also borrow some bit-packing tricks from the Yjs CRDT library \cite{yjs}.

We first topologically sort the events in the graph. Different replicas may sort the graph differently, but locally to one replica we can identify an event by its index in this sorted order.
Then we store different properties of events in separate byte sequences called \emph{columns}, which are then combined into one file with a simple header.
Each column stores different fields of the event data. The columns are:

\begin{itemize}
\item \emph{Event type, start position, and run length.} For example, ``the first 23 events are insertions at consecutive indexes starting from index 0, the next 10 events are deletions at consecutive indexes starting from index 7,'' and so on. We encode this using a variable-length binary encoding of integers, which represents small numbers in one byte, larger numbers in two bytes, etc.
\item \emph{Inserted content.} An insertion event contains exactly one character (a Unicode scalar value), and a deletion does not. We concatenate the UTF-8 encoding of the characters for insertion events in the same order as they appear in the first column, and LZ4-compress.
\item \emph{Parents.} By default we assume that every event has exactly one parent, namely its predecessor in the topological sort. Any events for which this is not true are listed explicitly, for example: ``the first event has zero parents; the 153rd event has two parents, namely events 31 and 152;'' and so on.
\item \emph{Event IDs.} Each event is uniquely identified by a pair of a replica ID and a per-replica sequence number. This column stores runs of event IDs, for example: ``the first 1085 events are from replica $A$, starting with sequence number 0; the next 595 events are from replica $B$, starting with sequence number 0;'' and so on.
\end{itemize}


Replicas can optionally also store a copy of the final document state reflecting all events. This allows documents to be loaded from disk without replaying the event graph.

We send the same data format over the network when replicating the entire event graph.
When sending a subset of events over the network (e.g., a single event during real-time collaboration), references to parent events outside of that subset need to be encoded using event IDs of the form $(\mathit{replicaID}, \mathit{seqNo})$, but otherwise the encoding is similar.

\begin{table*}
  \caption{The text editing traces used in our evaluation.
    \emph{Repeats}: number of times the original trace was repeated to normalise its length relative to the other traces.
    \emph{Events}: total number of editing events, in thousands, including repeats. Each inserted or deleted character counts as one event.
    \emph{Average concurrency}: mean number of concurrent branches per event in the trace.
    \emph{Graph runs}: number of sequential runs of events (linear event sequences without branching/merging).
    \emph{Authors}: number of users who added at least one event.
    \emph{Chars remaining}: percentage of inserted characters that remain in the document (i.e., are never deleted) after all events have been merged.
    \emph{Final size}: resulting document size in kilobytes after all events have been merged.
  }
  \label{traces-table}
  \footnotesize
\begin{tabular}{ccrrrrrrr}
\toprule
\textbf{Name} &
\textbf{Type} &
\textbf{Repeats} &
\textbf{Events (k)} &
\textbf{Avg Concurrency} &
\textbf{Graph runs} &
\textbf{Authors} &
\textbf{Chars remaining (\%)} &
\textbf{Final size (kB)} \\\midrule
S1 & sequential & 3 & 779 & 0.00 & 1 & 2 & 57.5 & 307.2 \\
S2 & sequential & 3 & 1105 & 0.00 & 1 & 1 & 26.7 & 166.3 \\
S3 & sequential & 1 & 2339 & 0.00 & 1 & 2 & 9.9 & 119.5 \\
C1 & concurrent & 25 & 652 & 0.43 & 92101 & 2 & 90.1 & 521.5 \\
C2 & concurrent & 25 & 608 & 0.44 & 133626 & 2 & 93.0 & 516.3 \\
A1 & asynchronous & 1 & 947 & 0.10 & 101 & 194 & 7.8 & 37.2 \\
A2 & asynchronous & 2 & 698 & 6.11 & 2430 & 299 & 49.6 & 222.0 \\
\bottomrule
\end{tabular}%
\end{table*}

\section{Evaluation}\label{benchmarking}


We created a TypeScript implementation of \algname optimised for simplicity and readability \cite{reference-reg}, and a production-ready Rust implementation optimised for performance \cite{dt}.
The TypeScript version omits the run-length encoding of internal state, B-trees, and topological sorting heuristics.

To evaluate the correctness of \algname we proved that the algorithm complies with Attiya et al.'s \emph{strong list specification} \cite{Attiya2016} (see \ifincludeappendix\autoref{proofs}\else the extended version of this paper \cite{extended-version}\fi).
We also performed randomised property testing on the implementations, including checking that our implementations converge to the same result.

\subsection{Editing traces}\label{editing-traces}

As there is no established benchmark for collaborative text editing, we collected a set of editing traces from real documents and made them freely available \cite{editing-traces}.
Statistics for these traces are given in \autoref{traces-table}.
The traces represent the editing history of the following documents:

\begin{description}
    \item[Sequential Traces:] These traces have no concurrency. Trace S1 is the LaTeX source of a journal paper \cite{Kleppmann2017,automerge-perf}, S2 is an 8,800-word blog post \cite{crdts-go-brrr}, and S3 is the text of this paper that you are currently reading.
        S2 has one author; S1 and S3 have two authors who took turns.
    \item[Concurrent Traces:] Trace C1 is two users collaboratively writing a reflection on TV series they have just watched. C2 is two users collaboratively reflecting on going to clown school together. We added 1~sec (C1) or 0.5~sec (C2) artificial latency between the users to increase the incidence of concurrent operations.
    \item[Asynchronous Traces:] We reconstructed the editing trace of some files in Git repositories. The event graph mirrors the branching/merging of Git commits. Since Git does not record individual keystrokes, we generated the minimal edit operations necessary to perform each commit's diff. Trace A1 is \texttt{src/node.cc} from the Git repository for Node.js \cite{node-src-nodecc}, and A2 is \texttt{Makefile} from the Git repository for Git itself \cite{git-makefile}.
\end{description}

Even though the sequential traces do not exercise the merging algorithm, they are important to include in the benchmark.
Anecdotal evidence suggests that the majority of documents in practice are sequentially edited~-- that is, they either have a single author, or multiple authors who take turns to write.
The concurrent traces have many short-lived branches (see the \emph{graph runs} column in \autoref{traces-table}).
The asynchronous traces have a small number of long-running branches, which occur in the context of offline working or editors that support explicit branching and merging~\cite{Upwelling,Patchwork}.

We recorded the sequential and concurrent traces ourselves, collaborating with friends or colleagues, using an instrumented text editor that recorded keystroke-granularity editing events.
All contributors to the traces have given their consent for their recorded keystroke data to be made publicly available and to be used for benchmarking purposes.
The asynchronous traces are derived from public data on GitHub.

The recorded editing traces originally varied a great deal in length.
To allow easier comparison of measurements between traces, we have normalised the length of the traces to contain approximately 500k inserted characters (except for S3, which is approximately twice this size).
We did this by repeating the original S1 and S2 traces 3 times, the original C1 and C2 traces 25 times, and the original A2 trace twice.
The statistics given in \autoref{traces-table} are after repetition.

\subsection{Experimental approach}

To evaluate the performance of \algname, we compare our Rust implementation with two popular CRDT libraries: Automerge v0.5.9 \cite{automerge} (Rust) and Yjs v13.6.10 \cite{yjs} (JavaScript).\footnote{We also tested Yrs \cite{yrs}, the Rust rewrite of Yjs by the original authors. At the time of our experiments it performed worse than Yjs, so we omitted it from our results.}
We only test their collaborative text datatypes, and not the other features they support.
However, the performance of these libraries varies widely.
In an effort to distinguish between implementation differences and algorithmic differences, we have also implemented our own performance-optimised reference CRDT library.
This library shares most of its code with our Rust \algname implementation, enabling a more like-to-like comparison between the traditional CRDT approach and \algname.
Our reference CRDT outperforms both Yjs and Automerge.

We have also implemented a simple OT library using the TTF algorithm \cite{Oster2006TTF}.
We do not use the server-based Jupiter algorithm \cite{Nichols1995} or the popular OT library ShareDB \cite{sharedb} because they do not support the branching and merging patterns that occur in our asynchronous traces.

We compare these implementations along 3 dimensions:\footnote{Experimental setup: We ran the benchmarks on a Ryzen 7950x CPU running Linux 6.5.0-28 and 64GB of RAM.
We compiled Rust code with rustc v1.78.0 in release mode with ``\texttt{-C target-cpu=native}''. Rust code was pinned to a single CPU core to reduce variance across runs. 
For JavaScript (Yjs) we used Node.js v22.2.0. 
All reported time measurements are the mean of at least 100 test iterations (except for the case where OT takes an hour to merge trace A2, which we ran 10 times).
The standard deviation for all benchmark results was less than 1.2\% of the mean, except for the Yjs measurements, which had a stddev of less than 6\%. Error bars on our graphs are too small to be visible.}

\begin{description}
\item[Speed:] The CPU time to load a document into memory, and to merge a set of updates from a remote replica.
\item[Memory usage:] The RAM used while a document is loaded and while merging remote updates.
\item[Storage size:] The number of bytes needed to persistently store a document or send it over the network.
\end{description}

By design, our experiments do not run over a network, but focus on single-node CPU and memory use (along with storage size, which is also a measure of network bandwidth used).
We made this choice because memory use and loading time are the biggest challenges with CRDTs, and CPU time on long-running branches is the biggest challenge with OT.
We delegate replication to the reliable broadcast protocol, which is beyond the scope of this paper.
Moreover, the results in a distributed setup would depend highly on the online/offline pattern we assume; running on a single node allows us to more directly compare the algorithms.


\subsection{Time taken to load and merge changes}

The slowest operations in many collaborative editors are:
\begin{itemize}
\item merging a large set of edits from a remote replica into the local state (e.g. reconnecting after working offline);
\item loading a document from disk into memory so that it can be displayed and edited.
\end{itemize}

To simulate a worst-case merge, we start with an empty document and then merge an entire editing trace into it.
In the case of \algname this means replaying the full trace.
\autoref{chart-remote} shows the merge time for each implementation.
Such a large merge occurs in practice when a node has been offline for a long time and needs to catch up on most of the editing history, or when replaying a subset of the editing history in order to reconstruct a historical version of the document (for history visualisation).
Moreover, in the CRDT implementations we tested, loading a document from disk takes the same CPU time as merging all of the events.
Smaller merges, which occur during online collaboration, are faster.

After merging the entire trace, we save the resulting local replica state to disk and measure the CPU time to load it back into memory.
Since loading and merging take the same time in the CRDTs we tested, we do not show their loading times separately in \autoref{chart-remote}.
In these algorithms, the CRDT metadata needs to be in memory for the user to be able to edit the document, or to apply any updates received from other replicas (even when there is no concurrency).
In contrast, OT and \algname can load documents orders of magnitude faster than CRDTs by caching the final document state on disk, and loading just this data (essentially a plain text file).
\algname and OT only need to load the event graph when merging concurrent changes or to reconstruct old document versions.
Document edits by the local user or applying non-concurrent remote events do not need the event graph.


\begin{figure}
  \includegraphics[width=\linewidth]{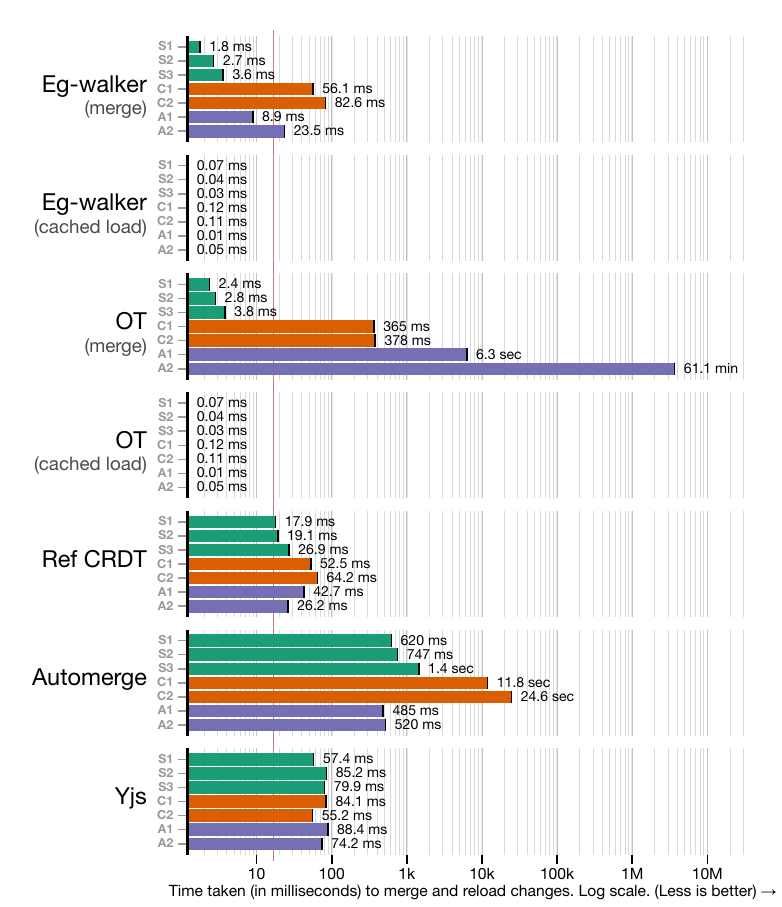}
  \caption{The CPU time taken by each algorithm to merge all events in each trace (as received from a remote replica), or to reload the resulting document from disk. The CRDT implementations (Ref CRDT, Automerge and Yjs) take the same amount of time to merge changes as they do to subsequently load the document. The red line at 16 ms indicates the time budget available to an application that wants to show the results of an operation by the next frame, assuming a display with a 60 Hz refresh rate.}
  \label{chart-remote}
\end{figure}

We can see in \autoref{chart-remote} that \algname and OT are very fast to merge the sequential traces (S1, S2, S3), since they simply apply the operations with no transformation.
However, OT performance degrades dramatically on the asynchronous traces (6 seconds for A1, and 1 hour for A2) due to the quadratic complexity of the algorithm, whereas \algname remains fast (160,000$\times$ faster in the case of A2).

On the concurrent traces (C1, C2) and asynchronous trace A2, the merge time of \algname is similar to that of our reference CRDT, since they perform similar work.
Both are significantly faster than the state-of-the-art Yjs and Automerge CRDT libraries; this is due to implementation differences and not fundamental algorithmic reasons.

On the sequential traces \algname outperforms our reference CRDT by a factor of 7--10$\times$, and on trace A1 (which contains large sequential sections) \algname is 5$\times$ faster.
Comparing to Yjs or Automerge, this speedup is greater still.
This is due to \algname's ability to clear its internal state and skip all of the internal state manipulation on critical versions (\autoref{clearing}).
To quantify this effect, \autoref{speed-ff} compares the time taken to replay all our traces with this optimisation enabled and disabled.
We see that the optimisation is effective for S1, S2, S3, and A1, whereas for C1, C2, and A2 it makes little difference (A2 contains no critical versions).

\begin{figure}
  \includegraphics[width=\linewidth]{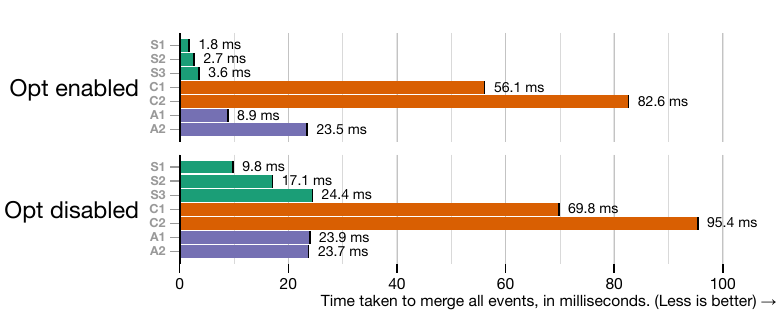}
  \caption{Time taken for \algname to merge all events in a trace, with and without the optimisations from \autoref{clearing}.}
  \label{speed-ff}
\end{figure}


When merging an event graph with very high concurrency (like A2), the performance of \algname is highly dependent on the order in which events are traversed.
A poorly chosen traversal order can make this trace as much as 8$\times$ slower to merge. Our topological sort algorithm (\autoref{graph-walk}) tries to avoid such pathological cases.

\subsection{RAM usage}

\autoref{chart-memusage} shows the memory footprint (retained heap size) of each algorithm.
For \algname and OT it shows both peak usage (while replaying the entire editing trace) and ``steady state'' memory usage (after temporary data and \algname's internal state are discarded and the event graph is written out to disk).
For the CRDTs the figure shows steady state memory usage; peak usage is up to 25\% higher.

\algname's peak memory use is similar to our reference CRDT's steady state: slightly lower on the sequential traces, and approximately double for the concurrent traces.
However, the steady-state memory use of \algname is 1--2 orders of magnitude lower than the best CRDT.
This is a significant result, since the steady state is what matters during normal operation while a document is being edited.
Memory usage reaches this peak only when replaying the entire trace; it is lower when merging a small branch.
Yjs has up to a 3$\times$ greater memory use than our reference CRDT, and Automerge an order of magnitude greater.


OT has the same memory use as \algname in the steady state, but significantly higher peak memory use on the C1, C2, and A2 traces (6.8~GiB for A2).
The reason is that our OT implementation memoizes intermediate transformed operations to improve performance.
This memory use could be reduced at the cost of increased merge times.
The computer we used for benchmarking had enough RAM to prevent swapping in all cases.

\begin{figure}
  \includegraphics[width=\linewidth]{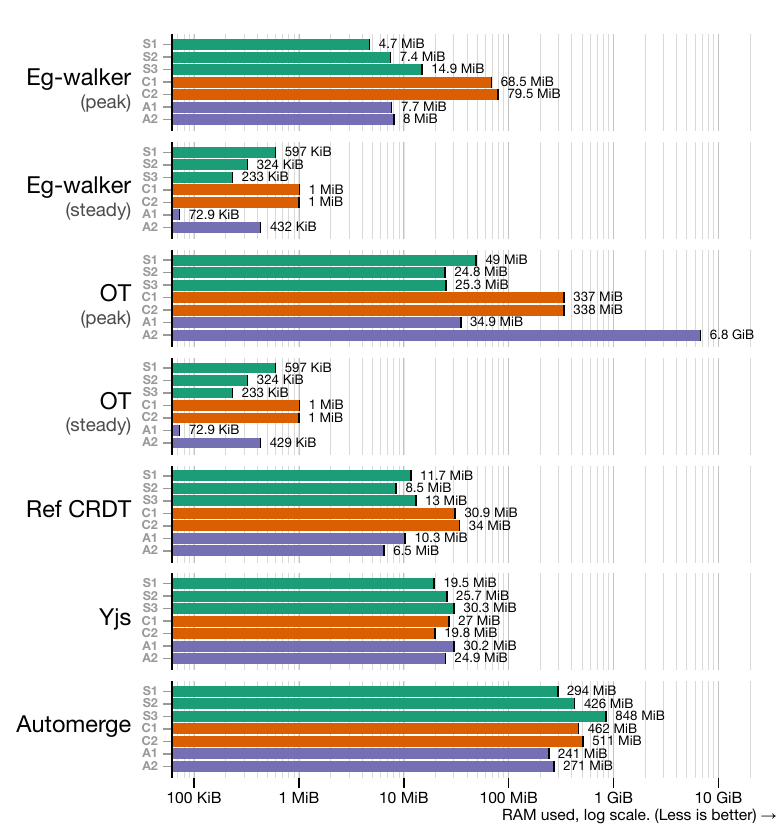}
  \caption{RAM used while merging an editing trace received from another replica. \algname and OT only retain the current document text in the steady state, but need additional RAM at peak while merging concurrent changes.}
  \label{chart-memusage}
\end{figure}

\subsection{Storage size}

Our binary encoding of event graphs (\autoref{storage}) results in smaller files than the equivalent internal CRDT state persisted by Automerge, and in many cases, Yjs.
To ensure a like-for-like comparison we have disabled \algname's built-in LZ4 and Automerge's built-in gzip compression. Enabling this compression further reduces the file sizes.


Automerge stores the full editing history of a document, and \autoref{chart-dt-vs-automerge} shows the resulting file sizes relative to the raw concatenated text content of all insertions, with and without a cached copy of the final document state (to enable fast loads).


In contrast, Yjs only stores the resulting document text, and any data needed to merge changes.
Yjs does not store deleted characters or the happened-before relationship between events.
\autoref{chart-dt-vs-yjs} compares Yjs to the equivalent event graph encoding in which we only store the final document text and operation metadata.
Our encoding is smaller than Yjs on the sequential and async traces, but larger for the concurrent traces, where the edges in the event graph take more space.
The overhead of storing the event graph is between 20\% and 3$\times$ the final plain text file size.

\begin{figure}
  \includegraphics[width=\linewidth]{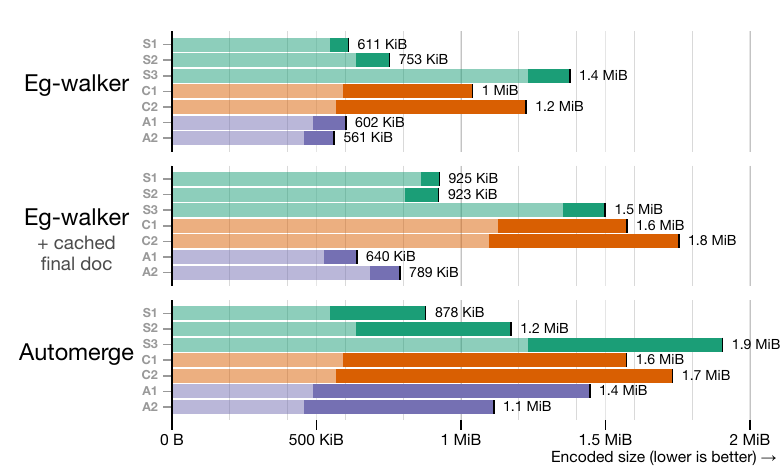}
  \caption{File size storing edit traces using \algname's event graph encoding (with and without final document caching) compared to Automerge. The lightly shaded region in each bar shows the concatenated length of all stored text. This acts as lower bound on the file size.}
  \label{chart-dt-vs-automerge}
\end{figure}

\begin{figure}
  \includegraphics[width=\linewidth]{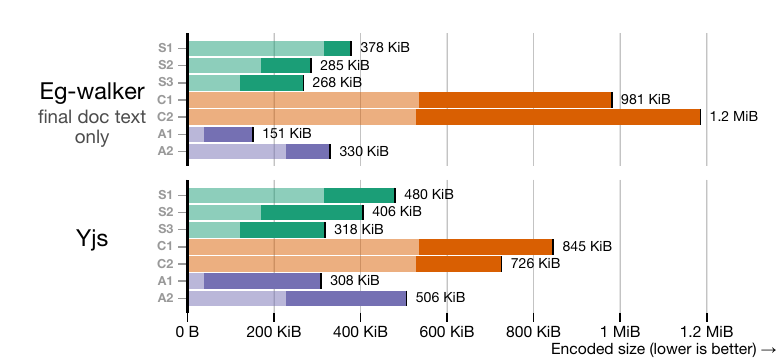}
  \caption{File size storing edit traces in which deleted text content has been omitted, as is the case with Yjs. The lightly shaded region in each bar is the size of the final document, which is a lower bound on the file size.}
  \label{chart-dt-vs-yjs}
\end{figure}

\section{Related Work}\label{related-work}

\algname is an example of a \emph{pure operation-based CRDT} \cite{polog}, which is a family of algorithms that capture a DAG (or partially ordered log) of operations in the form they were generated, and define the current state as a query over that log.
However, existing publications on pure operation-based CRDTs \cite{Almeida2023,Bauwens2023} present only datatypes such as maps, sets, and registers; \algname adds a list/text datatype to this family.

MRDTs \cite{Soundarapandian2022} are similarly based on a DAG, and use a three-way merge function to combine two branches since their lowest common ancestor; if the LCA is not unique, a recursive merge is used.
MRDTs for various datatypes have been defined, but so far none offers text with arbitrary insertion and deletion.

Toomim's \emph{time machines} approach \cite{time-machines} shares a conceptual foundation with \algname: both are based on traversing an event graph, with operations being transformed from their original form into a form that can be applied in topologically sorted order.
Toomim also points out that CRDTs can implement this transformation.
\algname is a concrete, optimised implementation of the time machine approach; novel contributions of \algname include updating the prepare version by retreating and advancing, as well as the details of internal state clearing and partial event graph replay.

\algname is also an \emph{operational transformation} (OT) algorithm \cite{Ellis1989}. 
OT has a long lineage of research going back to the 1990s \cite{Nichols1995,Ressel1996,Sun1998}.
To our knowledge, all existing OT algorithms consist of a set of \emph{transformation functions} that transform one operation with regard to one other operation, and a \emph{control algorithm} that traverses an editing history and invokes the necessary transformations.
A problem with this architecture is that when two replicas have diverged and each performed $n$ operations, merging their states unavoidably has a cost of at least $O(n^2)$; in some OT algorithms the cost is cubic or even worse \cite{Li2006,Roh2011RGA,Sun2020OT}.
\algname departs from the transformation function/control algorithm architecture and instead performs transformations using an internal CRDT state, which reduces the merging cost to $O(n \log n)$ in most cases; the upper bound of $O(n^2 \log n)$ is unlikely to occur in practical editing histories.


Other collaborative text editing algorithms \cite{Preguica2009,Roh2011RGA,fugue,Weiss2010} belong to the family of \emph{conflict-free replicated data types} (CRDTs) \cite{Shapiro2011}.
To our knowledge, all existing CRDTs for text work by assigning each character a unique ID, and translating index-based insertions and deletions into ID-based ones.
These unique IDs need to be held in memory when a document is being edited, persisted for the lifetime of the document, and sent to all replicas.
In contrast, \algname uses unique IDs only transiently during replay but does not persist or replicate them, and it can free its internal state whenever a critical version is reached.
\algname needs to store the event graph as long as concurrent operations may arrive, but this takes less space than CRDT state, and it only needs to be in-memory while merging concurrent operations.
Most of the time the event graph can remain on disk.

Gu et al.'s mark \& retrace method \cite{Gu2005} builds a CRDT-like structure containing the entire editing history, not only the parts being merged.
Differential synchronization \cite{Fraser2009} relies on heuristics such as similarity-matching of text to perform merges, which is not guaranteed to converge.

Version control systems such as Git \cite{Coglan2019}, Pijul \cite{pijul}, and Darcs \cite{darcs} also track the editing history of text files.
However, they do not support real-time collaboration, and they are line-based (good for code), whereas \algname is character-based (which is better for prose).
Git uses a three-way merge, which is not reliable on files containing substantial repeated text \cite{Khanna2007}.
Merges in Darcs have worst-case exponential complexity \cite{darcs-book}, and Pijul merges using a CRDT that assigns a unique ID to every line \cite{pijul-theory}.

\section{Conclusion}

\algname is a new approach to collaborative text editing that has characteristics of both CRDTs and OT.
It is orders of magnitude faster than existing algorithms in the best cases, and competitive with the fastest existing implementations in the worst cases.
Compared to existing CRDTs, it uses orders of magnitude less memory in the steady state, files are vastly faster to load for editing, and in documents with largely sequential editing edits from other users are merged much faster.
Compared to OT, large merges (e.g., when two users each did a significant amount of work while offline) are much faster, and \algname supports arbitrary branching/merging patterns (e.g., in peer-to-peer collaboration).

Since \algname stores a fine-grained editing history of a document, it allows applications to show that history to the user, and to restore arbitrary past versions of a document by replaying subsets of the graph.
The underlying event graph is not specific to the \algname algorithm, so we expect that the same data format will be able to support future collaborative editing algorithms as well.
The core idea of \algname is not specific to plain text; we believe it can be extended to other file types such as rich text, graphics, or spreadsheets.

Until now, many applications have been implemented using centralised server-based OT to avoid the overheads of CRDTs.
\algname is the first CRDT to match OT's memory use and performance on sequential editing histories (which are common in practice), while avoiding the quadratic merge complexity that makes OT impractical for long-running branches.
By requiring no server, \algname makes it possible for decentralised, local-first software \cite{Kleppmann2019localfirst} to become competitive with traditional cloud software.

\begin{acks}
  This work was made possible by generous support from Michael Toomim, the Braid community and the Invisible College. None of this would have happened without their help. Thank you for the endless conversations we have shared about collaborative editing.
  Martin Kleppmann gratefully acknowledges his crowdfunding supporters including Mintter and SoftwareMill.
  Thank you to Matthew Weidner, Joe Hellerstein, and our shepherd Diogo Behrens for feedback that helped improve this paper.
\end{acks}

\clearpage
\bibliographystyle{ACM-Reference-Format}
\bibliography{works}


\begin{thebibliography}{61}


\ifx \showCODEN    \undefined \def \showCODEN     #1{\unskip}     \fi
\ifx \showDOI      \undefined \def \showDOI       #1{#1}\fi
\ifx \showISBNx    \undefined \def \showISBNx     #1{\unskip}     \fi
\ifx \showISBNxiii \undefined \def \showISBNxiii  #1{\unskip}     \fi
\ifx \showISSN     \undefined \def \showISSN      #1{\unskip}     \fi
\ifx \showLCCN     \undefined \def \showLCCN      #1{\unskip}     \fi
\ifx \shownote     \undefined \def \shownote      #1{#1}          \fi
\ifx \showarticletitle \undefined \def \showarticletitle #1{#1}   \fi
\ifx \showURL      \undefined \def \showURL       {\relax}        \fi
\providecommand\bibfield[2]{#2}
\providecommand\bibinfo[2]{#2}
\providecommand\natexlab[1]{#1}
\providecommand\showeprint[2][]{arXiv:#2}

\bibitem[aut({[n.\,d.]})]%
        {automerge}
 \bibinfo{year}{[n.\,d.]}\natexlab{}.
\newblock \bibinfo{booktitle}{\emph{Automerge CRDT}}.
\newblock
\urldef\tempurl%
\url{https://automerge.org/}
\showURL{%
\tempurl}


\bibitem[dar({[n.\,d.]})]%
        {darcs}
 \bibinfo{year}{[n.\,d.]}\natexlab{}.
\newblock \bibinfo{booktitle}{\emph{Darcs}}.
\newblock
\urldef\tempurl%
\url{https://darcs.net/}
\showURL{%
\tempurl}


\bibitem[git({[n.\,d.]})]%
        {git-makefile}
 \bibinfo{year}{[n.\,d.]}\natexlab{}.
\newblock \bibinfo{booktitle}{\emph{Makefile for Git}}.
\newblock
\urldef\tempurl%
\url{https://github.com/git/git/blob/master/Makefile}
\showURL{%
\tempurl}


\bibitem[nod({[n.\,d.]})]%
        {node-src-nodecc}
 \bibinfo{year}{[n.\,d.]}\natexlab{}.
\newblock \bibinfo{booktitle}{\emph{Node.js source code: src/node.cc}}.
\newblock
\urldef\tempurl%
\url{https://github.com/nodejs/node/blob/main/src/node.cc}
\showURL{%
\tempurl}


\bibitem[pij({[n.\,d.]})]%
        {pijul-theory}
 \bibinfo{year}{[n.\,d.]}\natexlab{}.
\newblock \bibinfo{booktitle}{\emph{The Pijul manual: Theory}}.
\newblock
\urldef\tempurl%
\url{https://pijul.org/manual/theory.html}
\showURL{%
\tempurl}
\newblock
\shownote{Archived at \url{https://perma.cc/NAU4-SMYZ}}.


\bibitem[Abadi et~al\mbox{.}(2013)]%
        {Abadi2013}
\bibfield{author}{\bibinfo{person}{Daniel~J Abadi}, \bibinfo{person}{Peter
  Boncz}, \bibinfo{person}{Stavros Harizopoulos}, \bibinfo{person}{Stratos
  Idreos}, {and} \bibinfo{person}{Samuel Madden}.}
  \bibinfo{year}{2013}\natexlab{}.
\newblock \showarticletitle{The Design and Implementation of Modern
  Column-Oriented Database Systems}.
\newblock \bibinfo{journal}{\emph{Foundations and Trends in Databases}}
  \bibinfo{volume}{5}, \bibinfo{number}{3} (\bibinfo{year}{2013}),
  \bibinfo{pages}{197--280}.
\newblock
\urldef\tempurl%
\url{https://doi.org/10.1561/1900000024}
\showDOI{\tempurl}


\bibitem[Almeida(2023)]%
        {Almeida2023}
\bibfield{author}{\bibinfo{person}{Paulo~Sérgio Almeida}.}
  \bibinfo{year}{2023}\natexlab{}.
\newblock \showarticletitle{Approaches to Conflict-free Replicated Data Types}.
\newblock  (\bibinfo{date}{Oct.} \bibinfo{year}{2023}).
\newblock
\showeprint{2310.18220}
\urldef\tempurl%
\url{https://arxiv.org/abs/2310.18220}
\showURL{%
\tempurl}


\bibitem[Attiya et~al\mbox{.}(2016)]%
        {Attiya2016}
\bibfield{author}{\bibinfo{person}{Hagit Attiya}, \bibinfo{person}{Sebastian
  Burckhardt}, \bibinfo{person}{Alexey Gotsman}, \bibinfo{person}{Adam
  Morrison}, \bibinfo{person}{Hongseok Yang}, {and} \bibinfo{person}{Marek
  Zawirski}.} \bibinfo{year}{2016}\natexlab{}.
\newblock \showarticletitle{Specification and Complexity of Collaborative Text
  Editing}. In \bibinfo{booktitle}{\emph{ACM Symposium on Principles of
  Distributed Computing}} \emph{(\bibinfo{series}{PODC 2016})}.
  \bibinfo{pages}{259--268}.
\newblock
\urldef\tempurl%
\url{https://doi.org/10.1145/2933057.2933090}
\showDOI{\tempurl}


\bibitem[Baquero et~al\mbox{.}(2017)]%
        {polog}
\bibfield{author}{\bibinfo{person}{Carlos Baquero},
  \bibinfo{person}{Paulo~Sergio Almeida}, {and} \bibinfo{person}{Ali Shoker}.}
  \bibinfo{year}{2017}\natexlab{}.
\newblock \showarticletitle{Pure Operation-Based Replicated Data Types}.
\newblock  (\bibinfo{year}{2017}).
\newblock
\urldef\tempurl%
\url{https://arxiv.org/abs/1710.04469}
\showURL{%
\tempurl}


\bibitem[Bauwens and Gonzalez~Boix(2023)]%
        {Bauwens2023}
\bibfield{author}{\bibinfo{person}{Jim Bauwens} {and} \bibinfo{person}{Elisa
  Gonzalez~Boix}.} \bibinfo{year}{2023}\natexlab{}.
\newblock \showarticletitle{Nested Pure Operation-Based {CRDTs}}. In
  \bibinfo{booktitle}{\emph{37th European Conference on Object-Oriented
  Programming}} \emph{(\bibinfo{series}{ECOOP 2023})}.
  \bibinfo{publisher}{Schloss Dagstuhl}, \bibinfo{pages}{2:1--2:26}.
\newblock
\urldef\tempurl%
\url{https://doi.org/10.4230/LIPIcs.ECOOP.2023.2}
\showDOI{\tempurl}


\bibitem[Birman et~al\mbox{.}(1991)]%
        {Birman1991}
\bibfield{author}{\bibinfo{person}{Kenneth Birman}, \bibinfo{person}{Andr{\'e}
  Schiper}, {and} \bibinfo{person}{Pat Stephenson}.}
  \bibinfo{year}{1991}\natexlab{}.
\newblock \showarticletitle{Lightweight causal and atomic group multicast}.
\newblock \bibinfo{journal}{\emph{ACM Transactions on Computer Systems}}
  \bibinfo{volume}{9}, \bibinfo{number}{3} (\bibinfo{date}{Aug.}
  \bibinfo{year}{1991}), \bibinfo{pages}{272--314}.
\newblock
\urldef\tempurl%
\url{https://doi.org/10.1145/128738.128742}
\showDOI{\tempurl}


\bibitem[Boehm et~al\mbox{.}(1995)]%
        {Boehm1995}
\bibfield{author}{\bibinfo{person}{Hans-J. Boehm}, \bibinfo{person}{Russ
  Atkinson}, {and} \bibinfo{person}{Michael Plass}.}
  \bibinfo{year}{1995}\natexlab{}.
\newblock \showarticletitle{Ropes: An alternative to strings}.
\newblock \bibinfo{journal}{\emph{Software Prac. Experience}}
  \bibinfo{volume}{25}, \bibinfo{number}{12} (\bibinfo{year}{1995}),
  \bibinfo{pages}{1315--1330}.
\newblock
\urldef\tempurl%
\url{https://doi.org/10.1002/spe.4380251203}
\showDOI{\tempurl}


\bibitem[Cachin et~al\mbox{.}(2011)]%
        {Cachin2011}
\bibfield{author}{\bibinfo{person}{Christian Cachin}, \bibinfo{person}{Rachid
  Guerraoui}, {and} \bibinfo{person}{Luís Rodrigues}.}
  \bibinfo{year}{2011}\natexlab{}.
\newblock \bibinfo{booktitle}{\emph{Introduction to Reliable and Secure
  Distributed Programming} (\bibinfo{edition}{second} ed.)}.
\newblock \bibinfo{publisher}{Springer}.
\newblock
\showISBNx{9783642152597}


\bibitem[Coglan(2019)]%
        {Coglan2019}
\bibfield{author}{\bibinfo{person}{James Coglan}.}
  \bibinfo{year}{2019}\natexlab{}.
\newblock \bibinfo{booktitle}{\emph{Building {Git}}}.
\newblock
\urldef\tempurl%
\url{https://shop.jcoglan.com/building-git/}
\showURL{%
\tempurl}


\bibitem[Coldren(2024)]%
        {antarctica}
\bibfield{author}{\bibinfo{person}{Paul Coldren}.}
  \bibinfo{year}{2024}\natexlab{}.
\newblock \bibinfo{booktitle}{\emph{Engineering for Slow Internet: How to
  minimize user frustration in Antarctica}}.
\newblock
\urldef\tempurl%
\url{https://brr.fyi/posts/engineering-for-slow-internet}
\showURL{%
\tempurl}
\newblock
\shownote{Archived at \url{https://perma.cc/7Q9B-TUJV}}.


\bibitem[Cormen et~al\mbox{.}(2009)]%
        {CLRS2009}
\bibfield{author}{\bibinfo{person}{Thomas~H. Cormen},
  \bibinfo{person}{Charles~E. Leiserson}, \bibinfo{person}{Ronald~L. Rivest},
  {and} \bibinfo{person}{Clifford Stein}.} \bibinfo{year}{2009}\natexlab{}.
\newblock \bibinfo{booktitle}{\emph{Introduction to Algorithms}
  (\bibinfo{edition}{third} ed.)}.
\newblock \bibinfo{publisher}{MIT Press}.
\newblock


\bibitem[Day-Richter(2010)]%
        {DayRichter2010}
\bibfield{author}{\bibinfo{person}{John Day-Richter}.}
  \bibinfo{year}{2010}\natexlab{}.
\newblock \bibinfo{booktitle}{\emph{What's different about the new {Google
  Docs}: Making collaboration fast}}.
\newblock
\urldef\tempurl%
\url{https://drive.googleblog.com/2010/09/whats-different-about-new-google-docs.html}
\showURL{%
\tempurl}
\newblock
\shownote{Archived at \url{https://perma.cc/5FVM-542B}}.


\bibitem[Ellis and Gibbs(1989)]%
        {Ellis1989}
\bibfield{author}{\bibinfo{person}{C~A Ellis} {and} \bibinfo{person}{S~J
  Gibbs}.} \bibinfo{year}{1989}\natexlab{}.
\newblock \showarticletitle{Concurrency control in groupware systems}. In
  \bibinfo{booktitle}{\emph{ACM International Conference on Management of
  Data}} \emph{(\bibinfo{series}{SIGMOD 1989})}. \bibinfo{pages}{399--407}.
\newblock
\urldef\tempurl%
\url{https://doi.org/10.1145/67544.66963}
\showDOI{\tempurl}


\bibitem[Fabbri(2024)]%
        {ditto-military}
\bibfield{author}{\bibinfo{person}{Aaron Fabbri}.}
  \bibinfo{year}{2024}\natexlab{}.
\newblock \bibinfo{booktitle}{\emph{How Ditto Empowers the U.S. Navy's Unmanned
  Vehicle Operations, Keeping America Safe}}.
\newblock
\urldef\tempurl%
\url{https://ditto.live/blog/how-ditto-empowers-the-u-s-navy-s-unmanned-vehicle-operations-keeping-america-safe}
\showURL{%
\tempurl}
\newblock
\shownote{Archived at \url{https://perma.cc/UAM9-JX5C}}.


\bibitem[Fraser(2009)]%
        {Fraser2009}
\bibfield{author}{\bibinfo{person}{Neil Fraser}.}
  \bibinfo{year}{2009}\natexlab{}.
\newblock \showarticletitle{Differential synchronization}. In
  \bibinfo{booktitle}{\emph{9th ACM Symposium on Document Engineering}}
  \emph{(\bibinfo{series}{DocEng 2009})}. \bibinfo{publisher}{ACM},
  \bibinfo{pages}{13--20}.
\newblock
\urldef\tempurl%
\url{https://doi.org/10.1145/1600193.1600198}
\showDOI{\tempurl}


\bibitem[Gentle(2014)]%
        {sharedb}
\bibfield{author}{\bibinfo{person}{Joseph Gentle}.}
  \bibinfo{year}{2014}\natexlab{}.
\newblock \bibinfo{booktitle}{\emph{ShareDB}}.
\newblock
\urldef\tempurl%
\url{https://github.com/share/sharedb}
\showURL{%
\tempurl}


\bibitem[Gentle(2021)]%
        {crdts-go-brrr}
\bibfield{author}{\bibinfo{person}{Joseph Gentle}.}
  \bibinfo{year}{2021}\natexlab{}.
\newblock \bibinfo{booktitle}{\emph{5000x faster CRDTs: An Adventure in
  Optimization}}.
\newblock
\urldef\tempurl%
\url{https://josephg.com/blog/crdts-go-brrr/}
\showURL{%
\tempurl}


\bibitem[Gentle(2023a)]%
        {editing-traces}
\bibfield{author}{\bibinfo{person}{Joseph Gentle}.}
  \bibinfo{year}{2023}\natexlab{a}.
\newblock \bibinfo{booktitle}{\emph{Editing Traces (github repository)}}.
\newblock
\urldef\tempurl%
\url{https://github.com/josephg/editing-traces}
\showURL{%
\tempurl}


\bibitem[Gentle(2023b)]%
        {reference-reg}
\bibfield{author}{\bibinfo{person}{Joseph Gentle}.}
  \bibinfo{year}{2023}\natexlab{b}.
\newblock \bibinfo{booktitle}{\emph{Reference Eg-walker implementation in
  Typescript}}.
\newblock
\urldef\tempurl%
\url{https://github.com/josephg/eg-walker-reference}
\showURL{%
\tempurl}


\bibitem[Gentle(2024)]%
        {dt}
\bibfield{author}{\bibinfo{person}{Joseph Gentle}.}
  \bibinfo{year}{2024}\natexlab{}.
\newblock \bibinfo{booktitle}{\emph{Diamond Types: A fully featured realtime
  editing library}}.
\newblock
\urldef\tempurl%
\url{https://github.com/josephg/diamond-types}
\showURL{%
\tempurl}


\bibitem[Gomes et~al\mbox{.}(2017)]%
        {Gomes2017verifying}
\bibfield{author}{\bibinfo{person}{Victor B~F Gomes}, \bibinfo{person}{Martin
  Kleppmann}, \bibinfo{person}{Dominic~P Mulligan}, {and}
  \bibinfo{person}{Alastair~R Beresford}.} \bibinfo{year}{2017}\natexlab{}.
\newblock \showarticletitle{Verifying strong eventual consistency in
  distributed systems}.
\newblock \bibinfo{journal}{\emph{Proceedings of the ACM on Programming
  Languages (PACMPL)}} \bibinfo{volume}{1}, \bibinfo{number}{OOPSLA}
  (\bibinfo{date}{Oct.} \bibinfo{year}{2017}).
\newblock
\urldef\tempurl%
\url{https://doi.org/10.1145/3133933}
\showDOI{\tempurl}
\showeprint{1707.01747}


\bibitem[Good and Jeffery({[n.\,d.]})]%
        {automerge-storage}
\bibfield{author}{\bibinfo{person}{Alex Good} {and} \bibinfo{person}{Andrew
  Jeffery}.} \bibinfo{year}{[n.\,d.]}\natexlab{}.
\newblock \bibinfo{booktitle}{\emph{Automerge Binary Document Format}}.
\newblock
\urldef\tempurl%
\url{https://automerge.org/automerge-binary-format-spec/}
\showURL{%
\tempurl}
\newblock
\shownote{Archived at \url{https://perma.cc/XV25-RA4U}}.


\bibitem[Gu et~al\mbox{.}(2005)]%
        {Gu2005}
\bibfield{author}{\bibinfo{person}{Ning Gu}, \bibinfo{person}{Jiangming Yang},
  {and} \bibinfo{person}{Qiwei Zhang}.} \bibinfo{year}{2005}\natexlab{}.
\newblock \showarticletitle{Consistency maintenance based on the mark \&
  retrace technique in groupware systems}. In \bibinfo{booktitle}{\emph{ACM
  International Conference on Supporting Group Work}}
  \emph{(\bibinfo{series}{GROUP 2005})}. \bibinfo{publisher}{ACM},
  \bibinfo{pages}{264--273}.
\newblock
\urldef\tempurl%
\url{https://doi.org/10.1145/1099203.1099250}
\showDOI{\tempurl}


\bibitem[Hellerstein(2010)]%
        {Hellerstein2010}
\bibfield{author}{\bibinfo{person}{Joseph~M Hellerstein}.}
  \bibinfo{year}{2010}\natexlab{}.
\newblock \showarticletitle{The Declarative Imperative: Experiences and
  Conjectures in Distributed Logic}.
\newblock \bibinfo{journal}{\emph{ACM SIGMOD Record}} \bibinfo{volume}{39},
  \bibinfo{number}{1} (\bibinfo{date}{Sept.} \bibinfo{year}{2010}),
  \bibinfo{pages}{5--19}.
\newblock
\urldef\tempurl%
\url{https://doi.org/10.1145/1860702.1860704}
\showDOI{\tempurl}


\bibitem[Jahns({[n.\,d.]})]%
        {yjs}
\bibfield{author}{\bibinfo{person}{Kevin Jahns}.}
  \bibinfo{year}{[n.\,d.]}\natexlab{}.
\newblock \bibinfo{booktitle}{\emph{Yjs Shared Editing}}.
\newblock
\urldef\tempurl%
\url{https://yjs.dev/}
\showURL{%
\tempurl}


\bibitem[Khanna et~al\mbox{.}(2007)]%
        {Khanna2007}
\bibfield{author}{\bibinfo{person}{Sanjeev Khanna}, \bibinfo{person}{Keshav
  Kunal}, {and} \bibinfo{person}{Benjamin~C Pierce}.}
  \bibinfo{year}{2007}\natexlab{}.
\newblock \showarticletitle{A Formal Investigation of Diff3}. In
  \bibinfo{booktitle}{\emph{27th International Conference on Foundations of
  Software Technology and Theoretical Computer Science}}
  \emph{(\bibinfo{series}{FSTTCS 2007})}. \bibinfo{publisher}{Springer},
  \bibinfo{pages}{485--496}.
\newblock
\urldef\tempurl%
\url{https://doi.org/10.1007/978-3-540-77050-3_40}
\showDOI{\tempurl}


\bibitem[Kleppmann(2019)]%
        {automerge-columnar}
\bibfield{author}{\bibinfo{person}{Martin Kleppmann}.}
  \bibinfo{year}{2019}\natexlab{}.
\newblock \bibinfo{booktitle}{\emph{Experiment: columnar data encoding for
  Automerge}}.
\newblock
\urldef\tempurl%
\url{https://github.com/automerge/automerge-perf/blob/master/columnar/README.md}
\showURL{%
\tempurl}
\newblock
\shownote{Archived at \url{https://perma.cc/57KC-PP4Y}}.


\bibitem[Kleppmann(2020)]%
        {automerge-perf}
\bibfield{author}{\bibinfo{person}{Martin Kleppmann}.}
  \bibinfo{year}{2020}\natexlab{}.
\newblock \bibinfo{booktitle}{\emph{Benchmarking resources for Automerge}}.
\newblock
\urldef\tempurl%
\url{https://github.com/automerge/automerge-perf}
\showURL{%
\tempurl}


\bibitem[Kleppmann and Beresford(2017)]%
        {Kleppmann2017}
\bibfield{author}{\bibinfo{person}{Martin Kleppmann} {and}
  \bibinfo{person}{Alastair~R Beresford}.} \bibinfo{year}{2017}\natexlab{}.
\newblock \showarticletitle{A Conflict-Free Replicated {JSON} Datatype}.
\newblock \bibinfo{journal}{\emph{IEEE Transactions on Parallel and Distributed
  Systems}} \bibinfo{volume}{28}, \bibinfo{number}{10} (\bibinfo{date}{April}
  \bibinfo{year}{2017}), \bibinfo{pages}{2733--2746}.
\newblock
\urldef\tempurl%
\url{https://doi.org/10.1109/TPDS.2017.2697382}
\showDOI{\tempurl}
\showeprint{1608.03960}


\bibitem[Kleppmann et~al\mbox{.}({[n.\,d.]})]%
        {crdt-papers}
\bibfield{author}{\bibinfo{person}{Martin Kleppmann}, \bibinfo{person}{Annette
  Bieniusa}, {and} \bibinfo{person}{Marc Shapiro}.}
  \bibinfo{year}{[n.\,d.]}\natexlab{}.
\newblock \bibinfo{booktitle}{\emph{CRDT Papers}}.
\newblock
\urldef\tempurl%
\url{https://crdt.tech/papers.html}
\showURL{%
\tempurl}


\bibitem[Kleppmann et~al\mbox{.}(2019)]%
        {Kleppmann2019localfirst}
\bibfield{author}{\bibinfo{person}{Martin Kleppmann}, \bibinfo{person}{Adam
  Wiggins}, \bibinfo{person}{Peter van Hardenberg}, {and} \bibinfo{person}{Mark
  McGranaghan}.} \bibinfo{year}{2019}\natexlab{}.
\newblock \showarticletitle{Local-First Software: You own your data, in spite
  of the cloud}. In \bibinfo{booktitle}{\emph{ACM SIGPLAN International
  Symposium on New Ideas, New Paradigms, and Reflections on Programming and
  Software}} \emph{(\bibinfo{series}{Onward! 2019})}. \bibinfo{publisher}{ACM},
  \bibinfo{pages}{154--178}.
\newblock
\urldef\tempurl%
\url{https://doi.org/10.1145/3359591.3359737}
\showDOI{\tempurl}


\bibitem[Kow({[n.\,d.]})]%
        {darcs-book}
\bibfield{author}{\bibinfo{person}{Eric Kow}.}
  \bibinfo{year}{[n.\,d.]}\natexlab{}.
\newblock \bibinfo{booktitle}{\emph{Understanding Darcs}}.
\newblock
\urldef\tempurl%
\url{https://en.wikibooks.org/wiki/Understanding_Darcs/Print_Version}
\showURL{%
\tempurl}
\newblock
\shownote{Archived at \url{https://perma.cc/3VZF-8J65}}.


\bibitem[Lamport(1978)]%
        {Lamport1978}
\bibfield{author}{\bibinfo{person}{Leslie Lamport}.}
  \bibinfo{year}{1978}\natexlab{}.
\newblock \showarticletitle{Time, clocks, and the ordering of events in a
  distributed system}.
\newblock \bibinfo{journal}{\emph{Commun. ACM}} \bibinfo{volume}{21},
  \bibinfo{number}{7} (\bibinfo{year}{1978}), \bibinfo{pages}{558--565}.
\newblock
\urldef\tempurl%
\url{https://doi.org/10.1145/359545.359563}
\showDOI{\tempurl}


\bibitem[Li and Li(2006)]%
        {Li2006}
\bibfield{author}{\bibinfo{person}{Du Li} {and} \bibinfo{person}{Rui Li}.}
  \bibinfo{year}{2006}\natexlab{}.
\newblock \showarticletitle{A performance study of group editing algorithms}.
  In \bibinfo{booktitle}{\emph{12th International Conference on Parallel and
  Distributed Systems}} \emph{(\bibinfo{series}{ICPADS 2006})}.
\newblock
\urldef\tempurl%
\url{https://doi.org/10.1109/icpads.2006.18}
\showDOI{\tempurl}


\bibitem[Litt et~al\mbox{.}(2024)]%
        {Patchwork}
\bibfield{author}{\bibinfo{person}{Geoffrey Litt}, \bibinfo{person}{Paul
  Sonnentag}, \bibinfo{person}{Max Schöning}, \bibinfo{person}{Adam Wiggins},
  \bibinfo{person}{Peter van Hardenberg}, {and} \bibinfo{person}{Orion Henry}.}
  \bibinfo{year}{2024}\natexlab{}.
\newblock \bibinfo{booktitle}{\emph{Patchwork lab notebook: Version control for
  everything}}.
\newblock \bibinfo{type}{{T}echnical {R}eport}. \bibinfo{institution}{Ink \&
  Switch}.
\newblock
\urldef\tempurl%
\url{https://www.inkandswitch.com/patchwork/notebook/}
\showURL{%
\tempurl}
\newblock
\shownote{Archived at \url{https://perma.cc/VJM7-YPJB}}.


\bibitem[Lyu(2018)]%
        {vscode-buffer}
\bibfield{author}{\bibinfo{person}{Peng Lyu}.} \bibinfo{year}{2018}\natexlab{}.
\newblock \bibinfo{booktitle}{\emph{Text Buffer Reimplementation}}.
\newblock
\urldef\tempurl%
\url{https://code.visualstudio.com/blogs/2018/03/23/text-buffer-reimplementation}
\showURL{%
\tempurl}
\newblock
\shownote{Archived at \url{https://perma.cc/V695-K7EL}}.


\bibitem[McKelvey et~al\mbox{.}(2023)]%
        {Upwelling}
\bibfield{author}{\bibinfo{person}{Karissa~Rae McKelvey},
  \bibinfo{person}{Scott Jenson}, \bibinfo{person}{Eileen Wagner},
  \bibinfo{person}{Blaine Cook}, {and} \bibinfo{person}{Martin Kleppmann}.}
  \bibinfo{year}{2023}\natexlab{}.
\newblock \bibinfo{booktitle}{\emph{Upwelling: Combining real-time
  collaboration with version control for writers}}.
\newblock \bibinfo{type}{{T}echnical {R}eport}. \bibinfo{institution}{Ink \&
  Switch}.
\newblock
\urldef\tempurl%
\url{https://www.inkandswitch.com/upwelling/}
\showURL{%
\tempurl}
\newblock
\shownote{Archived at \url{https://perma.cc/7TT8-X8S9}}.


\bibitem[Meunier and Becker({[n.\,d.]})]%
        {pijul}
\bibfield{author}{\bibinfo{person}{Pierre-Étienne Meunier} {and}
  \bibinfo{person}{Florent Becker}.} \bibinfo{year}{[n.\,d.]}\natexlab{}.
\newblock \bibinfo{booktitle}{\emph{Pijul}}.
\newblock
\urldef\tempurl%
\url{https://pijul.org/}
\showURL{%
\tempurl}


\bibitem[Nichols et~al\mbox{.}(1995)]%
        {Nichols1995}
\bibfield{author}{\bibinfo{person}{David~A Nichols}, \bibinfo{person}{Pavel
  Curtis}, \bibinfo{person}{Michael Dixon}, {and} \bibinfo{person}{John
  Lamping}.} \bibinfo{year}{1995}\natexlab{}.
\newblock \showarticletitle{High-latency, low-bandwidth windowing in the
  {Jupiter} collaboration system}. In \bibinfo{booktitle}{\emph{8th Annual ACM
  Symposium on User Interface and Software Technology}}
  \emph{(\bibinfo{series}{UIST 1995})}. \bibinfo{pages}{111--120}.
\newblock
\urldef\tempurl%
\url{https://doi.org/10.1145/215585.215706}
\showDOI{\tempurl}


\bibitem[Nicolaescu et~al\mbox{.}(2016)]%
        {Nicolaescu2016YATA}
\bibfield{author}{\bibinfo{person}{Petru Nicolaescu}, \bibinfo{person}{Kevin
  Jahns}, \bibinfo{person}{Michael Derntl}, {and} \bibinfo{person}{Ralf
  Klamma}.} \bibinfo{year}{2016}\natexlab{}.
\newblock \showarticletitle{Near Real-Time Peer-to-Peer Shared Editing on
  Extensible Data Types}. In \bibinfo{booktitle}{\emph{19th International
  Conference on Supporting Group Work}} \emph{(\bibinfo{series}{GROUP 2016})}.
  \bibinfo{publisher}{ACM}, \bibinfo{pages}{39--49}.
\newblock
\urldef\tempurl%
\url{https://doi.org/10.1145/2957276.2957310}
\showDOI{\tempurl}


\bibitem[Oster et~al\mbox{.}(2006a)]%
        {Oster2006TTF}
\bibfield{author}{\bibinfo{person}{Gérald Oster}, \bibinfo{person}{Pascal
  Molli}, \bibinfo{person}{Pascal Urso}, {and} \bibinfo{person}{Abdessamad
  Imine}.} \bibinfo{year}{2006}\natexlab{a}.
\newblock \showarticletitle{Tombstone Transformation Functions for Ensuring
  Consistency in Collaborative Editing Systems}. In
  \bibinfo{booktitle}{\emph{9th IEEE International Conference on Collaborative
  Computing}} \emph{(\bibinfo{series}{CollaborateCom 2006})}.
\newblock
\urldef\tempurl%
\url{https://doi.org/10.1109/colcom.2006.361867}
\showDOI{\tempurl}


\bibitem[Oster et~al\mbox{.}(2006b)]%
        {Oster2006WOOT}
\bibfield{author}{\bibinfo{person}{Gérald Oster}, \bibinfo{person}{Pascal
  Urso}, \bibinfo{person}{Pascal Molli}, {and} \bibinfo{person}{Abdessamad
  Imine}.} \bibinfo{year}{2006}\natexlab{b}.
\newblock \showarticletitle{Data consistency for {P2P} collaborative editing}.
  In \bibinfo{booktitle}{\emph{ACM Conference on Computer Supported Cooperative
  Work}} \emph{(\bibinfo{series}{CSCW 2006})}. \bibinfo{pages}{259--268}.
\newblock
\urldef\tempurl%
\url{https://doi.org/10.1145/1180875.1180916}
\showDOI{\tempurl}


\bibitem[{Overleaf}({[n.\,d.]})]%
        {overleaf-ot}
\bibfield{author}{\bibinfo{person}{{Overleaf}}.}
  \bibinfo{year}{[n.\,d.]}\natexlab{}.
\newblock \bibinfo{booktitle}{\emph{Can multiple authors edit the same file at
  the same time?}}
\newblock
\urldef\tempurl%
\url{https://www.overleaf.com/learn/how-to/Can_multiple_authors_edit_the_same_file_at_the_same_time%3F}
\showURL{%
\tempurl}
\newblock
\shownote{Archived at \url{https://perma.cc/P8TH-YRMX}}.


\bibitem[Preguiça et~al\mbox{.}(2009)]%
        {Preguica2009}
\bibfield{author}{\bibinfo{person}{Nuno Preguiça},
  \bibinfo{person}{Joan~Manuel Marques}, \bibinfo{person}{Marc Shapiro}, {and}
  \bibinfo{person}{Mihai Letia}.} \bibinfo{year}{2009}\natexlab{}.
\newblock \showarticletitle{A Commutative Replicated Data Type for Cooperative
  Editing}. In \bibinfo{booktitle}{\emph{29th IEEE International Conference on
  Distributed Computing Systems}} \emph{(\bibinfo{series}{ICDCS 2009})}.
  \bibinfo{pages}{395--403}.
\newblock
\urldef\tempurl%
\url{https://doi.org/10.1109/icdcs.2009.20}
\showDOI{\tempurl}


\bibitem[Ratner(2024)]%
        {ditto-aircraft}
\bibfield{author}{\bibinfo{person}{Ryan Ratner}.}
  \bibinfo{year}{2024}\natexlab{}.
\newblock \bibinfo{booktitle}{\emph{ANA Elevates Onboard Passenger Experience
  with Ditto}}.
\newblock
\urldef\tempurl%
\url{https://ditto.live/blog/ana-elevates-onboard-passenger-experience-with-ditto}
\showURL{%
\tempurl}
\newblock
\shownote{Archived at \url{https://perma.cc/M6NG-AWLM}}.


\bibitem[Ressel et~al\mbox{.}(1996)]%
        {Ressel1996}
\bibfield{author}{\bibinfo{person}{Matthias Ressel}, \bibinfo{person}{Doris
  Nitsche-Ruhland}, {and} \bibinfo{person}{Rul Gunzenhäuser}.}
  \bibinfo{year}{1996}\natexlab{}.
\newblock \showarticletitle{An integrating, transformation-oriented approach to
  concurrency control and undo in group editors}. In
  \bibinfo{booktitle}{\emph{ACM Conference on Computer Supported Cooperative
  Work}} \emph{(\bibinfo{series}{CSCW 1996})}. \bibinfo{pages}{288--297}.
\newblock
\urldef\tempurl%
\url{https://doi.org/10.1145/240080.240305}
\showDOI{\tempurl}


\bibitem[Roh et~al\mbox{.}(2011)]%
        {Roh2011RGA}
\bibfield{author}{\bibinfo{person}{Hyun-Gul Roh}, \bibinfo{person}{Myeongjae
  Jeon}, \bibinfo{person}{Jin-Soo Kim}, {and} \bibinfo{person}{Joonwon Lee}.}
  \bibinfo{year}{2011}\natexlab{}.
\newblock \showarticletitle{Replicated Abstract Data Types: Building Blocks for
  Collaborative Applications}.
\newblock \bibinfo{journal}{\emph{J. Parallel and Distrib. Comput.}}
  \bibinfo{volume}{71}, \bibinfo{number}{3} (\bibinfo{date}{March}
  \bibinfo{year}{2011}), \bibinfo{pages}{354--368}.
\newblock
\showISSN{0743-7315}
\urldef\tempurl%
\url{https://doi.org/10.1016/j.jpdc.2010.12.006}
\showDOI{\tempurl}


\bibitem[Shapiro et~al\mbox{.}(2011)]%
        {Shapiro2011}
\bibfield{author}{\bibinfo{person}{Marc Shapiro}, \bibinfo{person}{Nuno
  Pregui\c{c}a}, \bibinfo{person}{Carlos Baquero}, {and} \bibinfo{person}{Marek
  Zawirski}.} \bibinfo{year}{2011}\natexlab{}.
\newblock \showarticletitle{Conflict-free Replicated Data Types}. In
  \bibinfo{booktitle}{\emph{13th International Conference on Stabilization,
  Safety, and Security of Distributed Systems}} \emph{(\bibinfo{series}{SSS
  2011})}. \bibinfo{pages}{386--400}.
\newblock
\urldef\tempurl%
\url{https://doi.org/10.1007/978-3-642-24550-3_29}
\showDOI{\tempurl}


\bibitem[Soundarapandian et~al\mbox{.}(2022)]%
        {Soundarapandian2022}
\bibfield{author}{\bibinfo{person}{Vimala Soundarapandian},
  \bibinfo{person}{Adharsh Kamath}, \bibinfo{person}{Kartik Nagar}, {and}
  \bibinfo{person}{KC Sivaramakrishnan}.} \bibinfo{year}{2022}\natexlab{}.
\newblock \showarticletitle{Certified mergeable replicated data types}. In
  \bibinfo{booktitle}{\emph{43rd ACM SIGPLAN International Conference on
  Programming Language Design and Implementation}} \emph{(\bibinfo{series}{PLDI
  2022})}. \bibinfo{publisher}{ACM}, \bibinfo{pages}{332--347}.
\newblock
\urldef\tempurl%
\url{https://doi.org/10.1145/3519939.3523735}
\showDOI{\tempurl}


\bibitem[Stonebraker et~al\mbox{.}(2005)]%
        {Stonebraker2005}
\bibfield{author}{\bibinfo{person}{Michael Stonebraker},
  \bibinfo{person}{Daniel~J Abadi}, \bibinfo{person}{Adam Batkin},
  \bibinfo{person}{Xuedong Chen}, \bibinfo{person}{Mitch Cherniack},
  \bibinfo{person}{Miguel Ferreira}, \bibinfo{person}{Edmond Lau},
  \bibinfo{person}{Amerson Lin}, \bibinfo{person}{Samuel Madden},
  \bibinfo{person}{Elizabeth O'Neil}, \bibinfo{person}{Patrick O'Neil},
  \bibinfo{person}{Alexander Rasin}, \bibinfo{person}{Nga Tran}, {and}
  \bibinfo{person}{Stanley Zdonik}.} \bibinfo{year}{2005}\natexlab{}.
\newblock \showarticletitle{{C-Store}: A Column-oriented {DBMS}}. In
  \bibinfo{booktitle}{\emph{31st International Conference on Very Large Data
  Bases}} \emph{(\bibinfo{series}{VLDB 2005})}. \bibinfo{pages}{553--564}.
\newblock


\bibitem[Sun et~al\mbox{.}(1998)]%
        {Sun1998}
\bibfield{author}{\bibinfo{person}{Chengzheng Sun}, \bibinfo{person}{Xiaohua
  Jia}, \bibinfo{person}{Yanchun Zhang}, \bibinfo{person}{Yun Yang}, {and}
  \bibinfo{person}{David Chen}.} \bibinfo{year}{1998}\natexlab{}.
\newblock \showarticletitle{Achieving convergence, causality preservation, and
  intention preservation in real-time cooperative editing systems}.
\newblock \bibinfo{journal}{\emph{ACM Transactions on Computer-Human
  Interaction}} \bibinfo{volume}{5}, \bibinfo{number}{1} (\bibinfo{date}{March}
  \bibinfo{year}{1998}), \bibinfo{pages}{63--108}.
\newblock
\urldef\tempurl%
\url{https://doi.org/10.1145/274444.274447}
\showDOI{\tempurl}


\bibitem[Sun et~al\mbox{.}(2020)]%
        {Sun2020OT}
\bibfield{author}{\bibinfo{person}{David Sun}, \bibinfo{person}{Chengzheng
  Sun}, \bibinfo{person}{Agustina Ng}, {and} \bibinfo{person}{Weiwei Cai}.}
  \bibinfo{year}{2020}\natexlab{}.
\newblock \showarticletitle{Real Differences between {OT} and {CRDT} in
  Correctness and Complexity for Consistency Maintenance in Co-Editors}.
\newblock \bibinfo{journal}{\emph{Proceedings of the ACM on Human-Computer
  Interaction}} \bibinfo{volume}{4}, \bibinfo{number}{CSCW1}, Article
  \bibinfo{articleno}{21} (\bibinfo{date}{May} \bibinfo{year}{2020}),
  \bibinfo{numpages}{30}~pages.
\newblock
\urldef\tempurl%
\url{https://doi.org/10.1145/3392825}
\showDOI{\tempurl}


\bibitem[Sypytkowski et~al\mbox{.}({[n.\,d.]})]%
        {yrs}
\bibfield{author}{\bibinfo{person}{Bartosz Sypytkowski}, \bibinfo{person}{Kevin
  Jahns}, {and} \bibinfo{person}{John Waidhofer}.}
  \bibinfo{year}{[n.\,d.]}\natexlab{}.
\newblock \bibinfo{booktitle}{\emph{Y CRDT: Rust port of Yjs}}.
\newblock
\urldef\tempurl%
\url{https://github.com/y-crdt/y-crdt}
\showURL{%
\tempurl}


\bibitem[Toomim(2024)]%
        {time-machines}
\bibfield{author}{\bibinfo{person}{Michael Toomim}.}
  \bibinfo{year}{2024}\natexlab{}.
\newblock \showarticletitle{CRDT and OT generalize as Time Machines}.
\newblock  (\bibinfo{year}{2024}).
\newblock
\urldef\tempurl%
\url{https://braid.org/time-machines}
\showURL{%
\tempurl}
\newblock
\shownote{Archived at \url{https://perma.cc/VND4-ACEK}}.


\bibitem[Weidner and Kleppmann(2023)]%
        {fugue}
\bibfield{author}{\bibinfo{person}{Matthew Weidner} {and}
  \bibinfo{person}{Martin Kleppmann}.} \bibinfo{year}{2023}\natexlab{}.
\newblock \showarticletitle{The Art of the Fugue: Minimizing Interleaving in
  Collaborative Text Editing}.
\newblock  (\bibinfo{year}{2023}).
\newblock
\urldef\tempurl%
\url{https://arxiv.org/abs/2305.00583}
\showURL{%
\tempurl}


\bibitem[Weiss et~al\mbox{.}(2010)]%
        {Weiss2010}
\bibfield{author}{\bibinfo{person}{Stéphane Weiss}, \bibinfo{person}{Pascal
  Urso}, {and} \bibinfo{person}{Pascal Molli}.}
  \bibinfo{year}{2010}\natexlab{}.
\newblock \showarticletitle{{Logoot-Undo}: Distributed Collaborative Editing
  System on {P2P} Networks}.
\newblock \bibinfo{journal}{\emph{IEEE Transactions on Parallel and Distributed
  Systems}} \bibinfo{volume}{21}, \bibinfo{number}{8} (\bibinfo{year}{2010}),
  \bibinfo{pages}{1162--1174}.
\newblock
\urldef\tempurl%
\url{https://doi.org/10.1109/tpds.2009.173}
\showDOI{\tempurl}


\end{thebibliography}

\clearpage
\appendix
\section{Artifact Appendix} 
\subsection{Abstract}

Our artifact is an implementation of the Eg-walker algorithm, together with the datasets and benchmarking tools required to run the experiments described in the paper.
Specifically, it contains the following items:
\begin{itemize}
    \item Editing traces of text documents that we use for benchmarking (see \autoref{traces-table}), in our own binary format and as JSON
    \item Source code for the following tools:
        \begin{itemize}
            \item Our optimised Eg-walker implementation, written in Rust
            \item Our reference CRDT implementation
            \item Our reference OT implementation
            \item Tools to convert our editing traces to JSON, and to the Yjs and Automerge file formats
            \item Benchmarking tools to reproduce all experiments
        \end{itemize}
    \item The tools to generate the figures in this paper
\end{itemize}

\subsection{Description \& Requirements}

\subsubsection{How to access}

All code and data of the artifact is publicly available in the following GitHub repository:

\url{https://github.com/josephg/egwalker-paper}

A snapshot of this repository is archived at:

\url{https://zenodo.org/records/13823409}

\doi{10.5281/zenodo.13823409}

The README file in this repository contains detailed instructions to configure \& run our code locally.

\subsubsection{Hardware dependencies}

To run our experiments, you need the following:

\begin{itemize}
    \item A computer running Linux. Our software should work on a number of other systems including Windows and MacOS, but we have not tested on other platforms.
    \item At least 8GB of RAM, but 16GB is recommended. This is mainly required for the \texttt{OT/A2} benchmark, which at peak uses approximately 7GB of RAM.
    \item Plenty of disk space. The Rust compiler produces a lot of temporary files~-- 44GB on our system.
\end{itemize}

\subsubsection{Software dependencies}

You need a recent compiler and runtime for the Rust programming language. We have tested with Rust 1.78; newer versions should also work.

You also need NodeJS installed in order to run the Yjs benchmark and to generate the charts showing our results. We used version 21.

Other software dependencies are managed through Cargo and npm respectively, and they are automatically installed as part of the build process.

\subsubsection{Benchmarks}

The datasets used by our benchmarks are included in the \texttt{datasets/} folder of our artifact repository, and the code to run the benchmarks is in the \texttt{tools/} folder.
Each of these folders contains a \texttt{README.md} file that documents its contents in more detail.

\subsection{Set-up}

Install Rust (the easiest way is via rustup) and NodeJS.

\subsection{Evaluation workflow}

Please see the \texttt{README.md} file in the artifact repository for a detailed description of the process.

\subsubsection{Major Claims}





\begin{description}
    \item[C1.] Eg-walker is competitive with existing state-of-the-art collaborative editing systems in terms of file size, memory usage, and CPU time taken to merge changes.
    That is, on all our benchmarks, Eg-walker has at most $\approx 2\times$ the cost of the other systems we test, and in some cases dramatically lower cost.
    Experiments E1, E2, E3 and Figures~\ref{chart-remote}--\ref{chart-dt-vs-yjs} support this claim.
    \item[C2.] On editing traces with long-running branches, the merge performance of Eg-walker is several orders of magnitude faster than OT, and slightly faster than the best CRDT implementations.
    This claim is supported by Experiment E3 and the data in \autoref{chart-remote}.
    \item[C3.] On sequential editing traces, Eg-walker is as fast as OT, and about an order of magnitude faster than our reference CRDT (which in turn is the fastest among the CRDTs we tested).
    This is supported by Experiment E3 and the data in \autoref{chart-remote}.
    This is in large part due to the optimisations in \autoref{clearing}, as shown in \autoref{speed-ff}.
    \item[C4.] The time to load a document from disk (to view and edit) with Eg-walker is the same as with OT, which is several orders of magnitude faster than all CRDTs we tested.
    Experiment E3 and \autoref{chart-remote} also show this.
    \item[C5.] The steady-state memory consumption of Eg-walker is at least an order of magnitude lower than that of CRDTs, and the same as OT.
    The peak memory consumption of Eg-walker during merging is similar to the steady state of CRDTs, and lower than OT.
    This claim is supported by Experiment E2 and the data in \autoref{chart-memusage}.
\end{description}

\subsubsection{Experiment E1: File sizes (4 compute-hours)}

Converts the raw editing traces in the \verb|datasets/raw| directory into the Yjs, Automerge, and Eg-walker file formats (\verb|datasets/*.{yjs,am,dt}|).
Some of the traces are repeated several times so that they have a similar size (see \autoref{editing-traces}).
The sizes of the resulting files are reported in \autoref{chart-dt-vs-automerge} and \autoref{chart-dt-vs-yjs}.
The files also form an input to the subsequent experiments.
The output files are included with the artifact so that the later experiments can be run without having to run this one.

To run this experiment:
\begin{verbatim}
  rm datasets/*
  ./step1-prepare.sh
\end{verbatim}
The resulting files should be byte-for-byte identical to those in the artifact.
This experiment is slow because we have not made much effort to optimise it.

To generate the figures:
\begin{verbatim}
  node collect.js
  cd svg-plot
  npm i # only needed once to install dependencies
  node render.js
\end{verbatim}
which writes \autoref{chart-dt-vs-automerge} to \verb|diagrams/filesize_full.svg| and \autoref{chart-dt-vs-yjs} to \verb|diagrams/filesize_smol.svg|.

\subsubsection{Experiment E2: Memory use (1 compute-hour)}

This experiment measures memory use~-- both the peak memory use while replaying each editing trace, and the ``steady state'' memory use once the replay is complete (keeping in memory the structures that are needed to display and edit a document, but freeing the structures that are needed only for replay and merging concurrent edits).
We test with the following algorithms:
\begin{itemize}
  \item Diamond-types (DT): our optimised Eg-walker implementation
  \item DT-CRDT: our reference CRDT implementation
  \item Automerge~\cite{automerge}
  \item Yjs~\cite{yjs}
  \item Yrs, the port of Yjs to Rust by the original authors~\cite{yrs}
  \item OT: our reference OT implementation
\end{itemize}

To run this experiment:
\begin{verbatim}
  ./step2a-memusage.sh
\end{verbatim}
Then use \verb|collect.js| and \verb|render.js| like in E1 to generate \autoref{chart-memusage} (\verb|diagrams/memusage.svg|).
The data appears in \verb|results/*_memusage.json|.

Note: Our OT implementation takes 1 hour to replay the A2 editing trace.

\subsubsection{Experiment E3: Merge time (12 compute-hours)}

This experiment measures the CPU time taken to merge the entire editing trace into an empty document, as if it had been received from a remote peer.
The benchmarks cover the same algorithms as Experiment E2.
To run this experiment:
\begin{verbatim}
  ./step2b-benchmarks.sh
\end{verbatim}
Then use \verb|collect.js| and \verb|render.js| like in E1 to generate the charts for \autoref{chart-remote} (\verb|diagrams/timings.svg|) and \autoref{speed-ff} (\verb|diagrams/ff.svg|).
The script writes summary data to \verb|results/timings.json|.

\subsection{Notes on Reusability}

While our artifact contains a snapshot of our Eg-walker implementation at the time of this paper was written, the ongoing development of the implementation is part of the Diamond Types project in the following repository:

\url{https://github.com/josephg/diamond-types}

Diamond Types is freely available under the ISC license.

As part of this work, we have also created several editing traces of real text documents, as described in \autoref{editing-traces}.
The intention is that these traces can be used in future research for benchmarking collaborative text editing systems.
The \verb|datasets/| directory of the artifact contains those traces in JSON format.
Additional traces that we may collect in the future will be added to the following repository:

\url{https://github.com/josephg/editing-traces}

That repository also contains documentation of the JSON-based file format that we use to encode the editing traces.

If you want to benchmark a CRDT using these editing traces, you need to convert them to your CRDT's local format. We do this by simulating (in memory) a set of collaborating peers. The peers fork and merge their changes. The \texttt{tools/crdt-converter} directory of our artifact contains code to perform this process using Automerge and Yjs (Yrs). We believe this algorithm could be adapted to support most existing CRDT formats and systems.

\subsection{General Notes}

The performance of collaborative editing systems varies by orders of magnitude depending on how the implementation has been optimised. For example, early CRDTs were widely thought to be impractical in real systems because they were so slow and memory-inefficient, taking gigabytes of RAM and hard disk space to process editing traces smaller than the ones we present in this paper. And yet, the CRDTs we benchmark here, like Yjs and Automerge, can process large documents with very reasonable computational resources.

These improvements have come through a mixture of improved data structures and algorithms, and improved implementation techniques such as reduced memory allocation and better memory layout.
Unfortunately, the techniques that these implementations use to achieve their performance are not well documented in either the academic literature or the documentation of those projects.
It is difficult to determine which combination of factors is responsible for a performance improvement.

As a result, it is difficult to fairly compare algorithms that are implemented by different developers and use different implementation techniques.
We have attempted to address this issue by writing our own reference CRDT and OT implementations that use a similar implementation style to our optimised Eg-walker implementation, enabling as much as possible an apples-for-apples comparison.
We hope that future research will be able to develop more rigorous approaches to evaluating collaborative editing systems.

\ifincludeappendix
\begin{listing*}
\section{Algorithm Pseudocode}\label{pseudocode-appendix}

\begin{flushleft}
\autoref{pseudocode1} and \autoref{pseudocode2} provide simplified pseudocode for the core \algname algorithm as described in \autoref{graph-walk} and \ref{prepare-effect-versions}.
For brevity, the pseudocode does not include the B-trees from \autoref{b-trees}, the state clearing optimisation from \autoref{clearing}, and the partial replay optimisation from \autoref{partial-replay}.
\end{flushleft}

\footnotesize
\vspace{5cm}
\begin{minted}{rust}
let events: EventStorage // Assumed to contain all events

enum PREPARE_STATE {
    NOT_YET_INSERTED = 0
    INSERTED = 1
    // Any state 2+ means the item has been concurrently deleted n-1 times.
}

// Each of these corresponds to a single inserted character.
type AugmentedCRDTItem {
    // The fields from the CRDT that determines insertion order
    id, originLeft, originRight,

    // State at effect version. Either inserted or inserted-and-subsequently-deleted.
    ever_deleted: bool,

    // State at prepare version (affected by retreat / advance)
    prepare_state: uint,
}

fn space_in_prepare_state(item: AugmentedCRDTItem) {
    if item.prepare_state == INSERTED { return 1 } else { return 0 }
}

fn space_in_effect_state(item: AugmentedCRDTItem) {
    if !item.ever_deleted { return 1 } else { return 0 }
}

// We have an efficient algorithm for this in our code. See diff() in causal-graph.ts.
fn diff(v1, v2) -> (only_in_v1, only_in_v2) {
    // This function considers the transitive expansion of the versions v1 and v2.
    // We return the set difference between the transitive expansions.
    let all_events_v1 = {set of all events in v1 + all events which happened-before any event in v1}
    let all_events_v2 = {set of all events in v2 + all events which happened-before any event in v2}

    return (
        set_subtract(all_events_v1 - all_events_v2),
        set_subtract(all_events_v2 - all_events_v1)
    )
}
\end{minted}
\caption{Pseudocode for the \algname algorithm (continued in \autoref{pseudocode2}).}\label{pseudocode1}
\end{listing*}

\begin{listing*}
\footnotesize
\begin{minted}{rust}
fn generateDocument(events) {
    let cur_version = {} // Frontier version

    // List of AugmentedCRDTItems. This could equally be an RGA tree or some other data structure.
    let crdt = []

    // Resulting document text
    let resulting_doc = ""

    // Some traversal obeying partial order relationship between events.
    for e in events.iter_in_causal_order() {
        // Step 1: Prepare
        let (a, b) = diff(cur_version, e.parent_version)
        for e in a {
            // Retreat
            let item = crdt.find_item_by_id(e.id)
            item.prepare_state -= 1
        }
        for e in b {
            // Advance
            let item = crdt.find_item_by_id(e.id)
            item.prepare_state += 1
        }

        // Step 2: Apply
        if e.type == Insert {
            // We find the insertion position in the crdt using the prepare_state variables.
            let ins_pos = idx_of(crdt, e.pos, PREPARE_STATE)
            // Then insert here using the underlying CRDT's rules.
            let origin_left = prev_item(ins_pos).id or START
            // Origin_right is the ID of the first item after ins_pos where prepare_state >= 1.
            let origin_right = next_item(crdt, ins_pos, item => item.prepare_state >= INSERTED).id or END

            // Use an existing CRDT to determine the order of concurrent insertions at the same position
            crdt_integrate(crdt, {
                id: e.id,
                origin_left,
                origin_right,
                ever_deleted: false,
                prepare_state: 1
            })

            let effect_pos = crdt[0..ins_pos].map(space_in_effect_state).sum()
            resulting_doc.splice_in(effect_pos, e.contents)
        } else {
            // Delete
            let idx = idx_of(crdt, e.pos, PREPARE_STATE)
            // But this time skip any items which aren't in the inserted state.
            while crdt[idx].prepare_state != INSERTED { idx += 1 }
            // Mark as deleted.
            crdt[idx].ever_deleted = true
            crdt[idx].prepare_state += 1

            let effect_pos = crdt[0..idx].map(space_in_effect_state).sum()
            resulting_doc.delete_at(effect_pos)
        }

        cur_version = {e.id}
    }

    return resulting_doc
}
\end{minted}
\caption{Continuation of the pseudocode in \autoref{pseudocode1}.}\label{pseudocode2}
\end{listing*}

\clearpage
\section{Proof of Correctness}\label{proofs}

We now demonstrate that \algname is correct by showing that it satisfies the \emph{strong list specification} proposed by Attiya et al. \cite{Attiya2016}, a formal specification of collaborative text editing.
Informally speaking, this specification requires that replicas converge to the same document state, that this state contains exactly those characters that were inserted but not deleted, and that inserted characters appear in the correct place relative to the characters that surrounded it at the time it was inserted.
Assuming network partitions are eventually repaired, this is a stronger specification than \emph{strong eventual consistency} \cite{Shapiro2011}, which is a standard correctness criterion for CRDTs \cite{Gomes2017verifying}.

With a suitable algorithm for ordering concurrent insertions at the same position, \algname is also able to achieve maximal non-interleaving \cite{fugue}, which is a further strengthening of the strong list specification.
However, since that algorithm is out of scope of this paper, we also leave the proof of non-interleaving out of scope.

\subsection{Definitions}

Let $\mathsf{Char}$ be the set of characters that can be inserted in a document.
Let $\mathsf{Op} = \{\mathit{Insert}(i, c) \mid i \in \mathbb{N} \wedge c \in \mathsf{Char}\} \cup \{\mathit{Delete}(i) \mid i \in \mathbb{N}\}$ be the set of possible operations.
Let $\mathsf{ID}$ be the set of unique event identifiers, and let $\mathsf{Evt} = \mathsf{ID} \times \mathcal{P}(\mathsf{ID}) \times \mathsf{Op}$ be the set of possible events consisting of a unique ID, a set of parent event IDs, and an operation.
When $e \in G$ and $e = (i,p,o)$ we also use the notation $e.\mathit{id} = i$, $e.\mathit{parents} = p$, and $e.\mathit{op} = o$.

\begin{definition}\label{valid-graph}
  An event graph $G \subseteq \mathsf{Evt}$ is \emph{valid} if:
  \begin{enumerate}
    \item every event $e \in G$ has an ID $e.\mathit{id}$ that is unique in $G$;
    \item for every event $e \in G$, every parent ID $p \in e.\mathit{parents}$ is the ID of some other event in $G$;
    \item the graph is acyclic, i.e. there is no subset of events $\{e_1, e_2, \dots, e_n\} \subseteq G$ such that $e_1$ is a parent of $e_2$, $e_2$ is a parent of $e_3$, \dots, and $e_n$ is a parent of $e_1$; and
    \item for every event $e \in G$, the index at which $e.\mathit{op}$ inserts or deletes is an index that exists (is not beyond the end of the document) in the document version defined by the parents $e.\mathit{parents}$.
  \end{enumerate}
\end{definition}

Since event graphs grow monotonically and we never remove events, it is easy to ensure that the graph remains valid whenever a new event is added to it.

Attiya et al. make a simplifying assumption that every insertion operation has a unique character.
We use a slightly stronger version of the specification that avoids this assumption.
We also simplify the specification by using our event graph definition instead of the original abstract execution definition (containing message broadcast/receive events and a visibility relation).
These changes do not affect the substance of the proof: each node of our event graph corresponds to a \emph{do} event in the original strong list specification, and the transitive closure of our event graph is equivalent to the visibility relation.

For a given event graph $G$ we define a replay function $\mathsf{replay}(G)$ as introduced in \autoref{replay}, based on the \algname algorithm.
It iterates over the events in $G$ in some topologically sorted order, transforming the operation in each event as described in \autoref{algorithm}, and then applying the transformed operation to the document state resulting from the operations applied so far (starting with the empty document).
In a real implementation, $\mathsf{replay}$ returns the final document state as a concatenated sequence of characters.
For the sake of this proof, we define $\mathsf{replay}$ to instead return a sequence of $(\mathit{id}, c)$ pairs, where $\mathit{id}$ is the unique ID of the event that inserted the character $c$.
This allows us to distinguish between different occurrences of the same character.
The text of the document can be recovered by simply ignoring the $\mathit{id}$ of each pair and concatenating the characters.

We can now state our modified definition of the strong list specification:

\begin{definition}\label{strong-list-spec}
  A collaborative text editing algorithm with a replay function $\mathsf{replay}(G)$ satisfies the \emph{strong list specification} if for every valid event graph $G \subset \mathsf{Evt}$ there exists a relation $\mathit{lo} \subset \mathsf{ID} \times \mathsf{ID}$ called the \emph{list order}, such that:
  \begin{enumerate}
    \item For event $e \in G$, let $G_e = \{e\} \cup \mathsf{Events}(e.\mathit{parents})$ be the set of all events that happened before $e$ and $e$ itself.
    Let $\mathit{doc}_e = \mathsf{replay}(G_e) = \langle (\mathit{id}_0, c_0), \dots, (\mathit{id}_{n-1}, c_{n-1}) \rangle$ be the document state immediately after locally generating $e$, where $c_i \in \mathsf{Char}$ and $\mathit{id}_i \in \mathsf{ID}$. Then:
      \begin{enumerate}
        \item $\mathit{doc}_e$ contains exactly the elements that have been inserted but not deleted in $G_e$:
         \begin{multline*}
             \hspace{9pt}(\exists i \in [0, n-1]: \mathit{doc}_e [i] = (\mathit{id}, c)) \Longleftrightarrow \\
             (\exists a \in G_e, j \in \mathbb{N}: a.\mathit{id} = \mathit{id} \wedge a.\mathit{op} = \mathit{Insert}(j,c)) \;\wedge \\
             (\nexists b \in G_e, k \in \mathbb{N}: b.\mathit{op} = \mathit{Delete}(k) \;\wedge \\
             \mathsf{replay}(\mathsf{Events}(b.\mathit{parents}))[k] = (\mathit{id}, c)).
         \end{multline*}
        \item The order of the elements in $\mathit{doc}_e$ is consistent with the list order:
          \begin{equation*}
            \forall i, j \in [0, n-1]: i<j \Longrightarrow (\mathit{id}_i, \mathit{id}_j) \in \mathit{lo}.
          \end{equation*}
        \item Elements are inserted at the specified position:
          \begin{equation*}
            \qquad\forall i, c: e.\mathit{op} = \mathit{Insert}(i,c) \Longrightarrow \mathit{doc}_e [i] = (e.\mathit{id}, c).
          \end{equation*}
      \end{enumerate}
    \item The list order $\mathit{lo}$ is transitive, irreflexive, and total, and thus determines the order of all insert operations in the event graph.
  \end{enumerate}
\end{definition}

\subsection{Proving Convergence}

\begin{lemma}\label{lemma-prepare-delete}
  Let $e$ be an event in a valid event graph such that $e.\mathit{op} = \mathit{Delete}(i)$.
  In the internal state immediately before applying $e$ (in which all events that happened before $e$ have been advanced and all others have been retreated), either the record that $e$ will update has $s_p = \texttt{Ins}$, or it is part of a placeholder (which behaves like a sequence of $s_p = \texttt{Ins}$ records).
\end{lemma}
\begin{proof}
  If we had $s_p = \texttt{NotInsertedYet}$, that would imply that we retreated the insertion of the character deleted by $e$, which contradicts the fact that the insertion of a character must happen before any deletion of the same character.
  Furthermore, if we had $s_p = \texttt{Del}\; k$ for some $k$, that would imply that an event that happened before $e$ already deleted the same character, in which case it would not be possible to generate $e$.
  This leaves $s_p = \texttt{Ins}$ or placeholder as the only options that do not result in a contradiction.
\end{proof}

\begin{lemma}\label{lemma-prepare-state}
  Let $S_0$ be some internal \algname state, and let $a$ and $b$ be two concurrent events.
  Let $S_1$ be the internal state resulting from updating $S_0$ with retreat and advance calls so that the prepare version of $S_1$ equals the parents of $b$.
  Let $S_2$ be the internal state resulting from first replaying $a$ on top of $S_0$, and then retreating and advancing so that the prepare version of $S_2$ equals the parents of $b$.
  Then the only difference between $S_1$ and $S_2$ is in the record inserted or updated by $a$ (and possibly the split of a placeholder that this record falls within); the rest of $S_1$ and $S_2$ is the same.
\end{lemma}
\begin{proof}
  Since $S_0$ is produced by \algname, it contains records for all characters that have been inserted or deleted by events since the last critical version prior to $a$ and $b$, it contains placeholders for any characters inserted but not deleted prior to that critical version, and it does not contain anything for characters that were deleted prior to that critical version.
  By the definition of critical version, any event $e$ that is concurrent with $a$ or $b$ must be after the critical version, and therefore the record that is updated by $e$ must exist in $S_0$.

  $S_1$ has the same record sequence and the same $s_e$ in each record as $S_0$, since retreating and advancing do not change those things.
  The $s_p$ values in $S_1$ are set so that every record inserted by an event that is concurrent with $b$ has $s_p = \texttt{NotInsertedYet}$, every record whose insertion happened before $b$ but which was not deleted before $b$ has $s_p = \texttt{Ins}$, and every record that was deleted by $k>0$ separate events before $b$ has $s_p = \texttt{Del}\; k$.
  To achieve this it is sufficient to consider events that happened after the last critical version.
  Thus, the $s_p$ values in $S_1$ do not depend on the $s_p$ values in $S_0$, and they do not depend on any events that are concurrent with $b$.

  Replaying $a$ on top of $S_0$ involves first updating the $s_p$ values to set the prepare version to the parents of $a$ (which may differ from the parents of $b$), and then applying $a$, which either inserts or updates a record in the internal state, and possibly splits a placeholder to accommodate this record.
  $S_2$ is then produced by updating all of the $s_p$ values in the same way as for $S_1$.
  As these $s_p$ values depend only on $b.\mathit{parents}$ and not on $a$, $S_2$ is identical to $S_1$ except for the record inserted or updated by $a$.
\end{proof}

\begin{lemma}\label{lemma-ins-ins}
  Let $a$ and $b$ be two concurrent events such that $a.\mathit{op} = \mathit{Insert}(i, c_i)$ and $b.\mathit{op} = \mathit{Insert}(j, c_j)$.
  If we start with some internal state and document state and then replay $a$ followed by $b$, the resulting internal state and document state are the same as if we had replayed $b$ followed by $a$.
\end{lemma}
\begin{proof}
  To replay $a$ followed by $b$, we first retreat/advance so that the prepare state corresponds to $a.\mathit{parents}$, then apply $a$, then retreat $a$, then retreat/advance so that the prepare state corresponds to $b.\mathit{parents}$, then apply $b$.
  Applying $a$ inserts a record into the internal state, and after retreating $a$ this record has $s_p = \texttt{NotInsertedYet}$ and $s_e = \texttt{Ins}$.
  Since $b$ is concurrent to $a$, $a$ cannot be a critical version, and therefore the internal state is not cleared after applying $a$.
  When $b$ is applied, the presence of the record inserted by $a$ is the only difference between the internal state when applying $b$ after $a$ compared to applying $b$ without applying $a$ first (by Lemma \ref{lemma-prepare-state}).
  When determining the insertion position in the internal state for $b$'s record based on $b$'s index $j$, the record inserted by $a$ does not count since it has $s_p = \texttt{NotInsertedYet}$.
  Therefore, $b$'s record is inserted into the internal state at the same position relative to its neighbours, regardless of whether $a$ has been applied previously.
  By similar argument the same holds for $a$'s record.

  As explained in \autoref{prepare-effect-versions}, the internal state uses a CRDT algorithm to place the records in the internal state in a consistent order, regardless of the order in which the events are applied.
  The details of that algorithm go beyond the scope of this paper.
  The key property of that algorithm is that the final sequence of internal state records is the same, regardless of whether we apply first $a$ and then $b$, or vice versa.
  For example, if we first apply $a$ then $b$, and if the final position of $b$'s record in the internal state is after $a$'s record, then the CRDT algorithm has to skip over $a$'s record (and potentially other, concurrently inserted records) when determining the insertion position for $b$'s record.
  This process never needs to skip over a placeholder, since placeholders represent characters that were inserted before the last critical version.
  It only ever needs to skip over records for insertions that are concurrent with $a$ or $b$; by the definition of critical versions, all such insertion events appear after the last critical version (and hence after the last internal state clearing) in the topological sort, and therefore they are represented by explicit internal state records, not placeholders.

  Now we consider the document state.
  WLOG assume that the record inserted by $a$ appears at an earlier position in the internal state than the record inserted by $b$ (regardless of the order of applying $a$ and $b$).
  Let $i'$ be the transformed index of $a.\mathit{op}$ when $a$ is applied first, and let $j'$ be the transformed index of $b.\mathit{op}$ when $b$ is applied first.

  Say we replay $a$ before $b$.
  When computing the transformed index for $b$, the internal state record for $a$ has $s_p = \texttt{NotInsertedYet}$, and hence it is not counted when mapping $b.\mathit{op}$'s index $j$ to $b$'s internal state record.
  However, $a$'s record \emph{is} counted when mapping $b$'s internal state record back to an index, since $a$'s record has $s_e = \texttt{Ins}$ and it appears before $b$'s record.
  Therefore the transformed index for $b.\mathit{op}$ is $j' + 1$ when applied after $a$.
  On the other hand, if we replay $b$ before $a$, the record for $b$ appears after the record for $a$ in the internal state, so the transformed index for $a$ is $i'$, unaffected by $b$.
  Thus, we have the situation as shown in \autoref{two-inserts}, and the effect of the two insertions $a$ and $b$ on the document state is the same regardless of their order.
\end{proof}

\begin{lemma}\label{lemma-ins-del}
  Let $a$ and $b$ be two concurrent events such that $a.\mathit{op} = \mathit{Insert}(i, c)$ and $b.\mathit{op} = \mathit{Delete}(j)$.
  If we start with some internal state and document state and then replay $a$ followed by $b$, the resulting internal state and document state are the same as if we had replayed $b$ followed by $a$.
\end{lemma}
\begin{proof}
  Since $a$ and $b$ are concurrent, the character being deleted by $b$ cannot be the character inserted by $a$.
  We therefore only need to consider two cases: (1)~the record inserted by $a$ has an earlier position in the internal state than the record updated by $b$; or (2) vice versa.

  Case (1): If we replay $a$ before $b$, we first apply $a$, then retreat $a$, then apply $b$ (and also retreat/advance other events before applying, like in Lemma \ref{lemma-ins-ins}).
  Applying $a$ inserts a record into the internal state, and after retreating $a$ this record has $s_p = \texttt{NotInsertedYet}$ and $s_e = \texttt{Ins}$.
  When subsequently applying $b$ we update an internal state record at a later position.
  The record inserted by $a$ is not counted when mapping $b$'s index to an internal record, but it is counted when mapping the internal record back to a transformed index, resulting in $b$'s transformed index being one greater than it would have been without earlier applying $a$ (by Lemma \ref{lemma-prepare-state}).
  On the other hand, if we replay $b$ before $a$, the record updated by $b$ appears after $a$'s record in the internal state, so the transformation of $a$ is not affected by $b$.
  The transformed operations therefore converge.

  Case (2): If we replay $b$ before $a$, we first apply $b$, then retreat $b$, then apply $a$ (plus other retreats/advances).
  Applying $b$ updates an existing record in the internal state (possibly splitting a placeholder in the process).
  Before applying $b$ this record must have $s_p = \texttt{Ins}$ (by Lemma \ref{lemma-prepare-delete}), and it can have either $s_e = \texttt{Ins}$ (in which case, the transformed operation for $b$ is $\mathit{Delete}(j')$ for some transformed index $j'$) or $s_e = \texttt{Del}$ (in which case, $b$ is transformed into a no-op).
  After applying and retreating $b$ this record has $s_p = \texttt{Ins}$ and $s_e = \texttt{Del}$ in any case.
  We next apply $a$, which by assumption inserts a record into the internal state at a later position than $b$'s record.
  If we had $s_e = \texttt{Del}$ before applying $b$, the process of applying and retreating $b$ did not change the internal state, so the transformed operation for $a$ is the same as if $b$ had not been applied, which is consistent with the fact that $b$ was transformed into a no-op.
  If we had $s_e = \texttt{Ins}$ before applying $b$, $b$'s record is counted when mapping $a$'s index to an internal record position, but not counted when mapping the internal record back to a transformed index, resulting in $a$'s transformed index being one less than it would have been without earlier applying $b$ (by Lemma \ref{lemma-prepare-state}), as required given that $b$ has deleted an earlier character.
  On the other hand, if we replay $a$ before $b$, the record inserted by $a$ appears after $b$'s record in the internal state, so the transformation of $b$ is not affected by $a$, and the transformed operations converge.
\end{proof}

\begin{lemma}\label{lemma-del-del}
  Let $a$ and $b$ be two concurrent events such that $a.\mathit{op} = \mathit{Delete}(i)$ and $b.\mathit{op} = \mathit{Delete}(j)$.
  If we start with some internal state and document state and then replay $a$ followed by $b$, the resulting internal state and document state are the same as if we had replayed $b$ followed by $a$.
\end{lemma}
\begin{proof}
  WLOG we need to consider two cases: (1)~the record updated by $a$ has an earlier position in the internal state than the record updated by $b$; or (2)~$a$ and $b$ update the same internal state record. The case where $a$'s record has a later position than $b$'s record is symmetric to (1).

  Case (1): We further consider two sub-cases: (1a)~the record that $a$ will update has $s_e = \texttt{Ins}$ prior to applying $a$; or (1b)~the record has $s_e = \texttt{Del}$.

  Case (1a): Say we replay $a$ before $b$.
  Before applying $a$, the record that $a$ will update must have $s_p = \texttt{Ins}$ (by Lemma \ref{lemma-prepare-delete}).
  After applying and retreating $a$, the record updated by $a$ has $s_p = \texttt{Ins}$ and $s_e = \texttt{Del}$, and the transformed operation for $a$ is $\mathit{Delete}(i')$ for some transformed index $i'$.
  We subsequently apply $b$, which by assumption updates an internal state record that is later than $a$'s.
  $a$'s record is therefore counted when mapping the index of $b.\mathit{op}$ to an internal record position, but not counted when mapping the internal record back to a transformed index.
  If $a$ had not been replayed previously, it would have been counted during both mappings (by Lemma \ref{lemma-prepare-state}).
  Thus, if the record updated by $b$ has $s_e = \texttt{Ins}$, the transformed operation for $b$ is $\mathit{Delete}(j'-1)$, where $j'$ is the transformed index of $b$'s operation if $a$ had not been replayed previously, and $j'-1 \geq i'$, as required.
  If $b$'s record previously has $s_e = \texttt{Del}$, it is transformed into a no-op.
  On the other hand, if we replay $b$ before $a$, the record updated by $b$ appears later than $a$'s record in the internal state, so the transformation of $a$ is not affected by $b$.

  Case (1b): Say we replay $a$ before $b$.
  Before applying $a$, the record that $a$ will update must have $s_p = \texttt{Ins}$ (by Lemma \ref{lemma-prepare-delete}), and $s_e = \texttt{Del}$ by assumption.
  After applying and retreating $a$, the record updated by $a$ remains in the same state ($s_p = \texttt{Ins}$, $s_e = \texttt{Del}$), and the transformed operation for $a$ is a no-op.
  When we subsequently apply $b$, the transformed operation is therefore the same as if $a$ had not been applied, as required.
  On the other hand, if we replay $b$ before $a$, the record updated by $b$ appears later than $a$'s record in the internal state, so the transformation of $a$ is not affected by $b$.

  Case (2): Before replaying both of the events, the record that both events update may have $s_e = \texttt{Ins}$ or $s_e = \texttt{Del}$, but after applying the first event it definitely has $s_e = \texttt{Del}$.
  The second event will therefore be transformed into a no-op.
  The same happens regardless of whether $a$ or $b$ is replayed first, so the result does not depend on the order of replay of the two events.
\end{proof}

\begin{lemma}\label{lemma-deterministic}
  Given a valid event graph $G$, $\mathsf{replay}(G)$ is a deterministic function.
  In other words, any two replicas that have the same event graph converge to the same document state and the same internal state.
\end{lemma}
\begin{proof}
  The algorithms to transform an operation and to apply a transformed operation to the document state are by definition deterministic.
  This leaves as the only source of nondeterminism the choice of topologically sorted order ($G$ is valid and hence acyclic, thus at least one such order exists, but there may be several topologically sorted orders if $G$ contains concurrent events).
  We show that all sort orders result in the same final document state.

  Let $E = \langle e_1, e_2, \dots, e_n \rangle$ and $E' = \langle e'_1, e'_2, \dots, e'_n \rangle$ be two topological sort orders of $G = \{e_1, e_2, \dots, e_n\}$.
  Then $E'$ must be a permutation of $E$.
  Both sequences are in some causal order, that is: if $e_i \rightarrow e_j$ ($e_i$ happens before $e_j$, as defined in \autoref{event-graphs}), then $e_i$ must appear before $e_j$ in both $E$ and $E'$.
  If $e_i \parallel e_j$ (they are concurrent), the events could appear in either order.
  Therefore, it is possible to transform $E$ into $E'$ by repeatedly swapping two concurrent events that are adjacent in the sequence.
  We show that at each such swap we maintain the invariant that the document state and the internal state resulting from replaying the events in the order before the swap are equal to the states resulting from replaying the events in the order after the swap.
  Therefore, the document state and the internal state resulting from replaying $E$ are equal to those resulting from $E'$.

  Let $\langle e_1, e_2, \dots, e_i, e_{i+1}, \dots, e_n \rangle$ be the sequence of events prior to one of these swaps, and $e_i$, $e_{i+1}$ are the events to be swapped.
  Replaying the events in the prefix $\langle e_1, e_2, \dots, e_{i-1} \rangle$ is a deterministic algorithm resulting in some document state and some internal state.
  Next, we replay either $e_i$ followed by $e_{i+1}$, or $e_{i+1}$ followed by $e_i$.
  Since $e_i$ and $e_{i+1}$ are concurrent, it is not possible for only one of the two to be contained in a critical version, and therefore no state clearing will take place between applying these two events.
  If $e_i$ and $e_{i+1}$ are both insertions, the resulting states in either order are the same by Lemma \ref{lemma-ins-ins}.
  If one of $e_i$ and $e_{i+1}$ is an insertion and the other is a deletion, we use Lemma \ref{lemma-ins-del}.
  If both $e_i$ and $e_{i+1}$ are deletions, we use Lemma \ref{lemma-del-del}.
  Finally, replaying the suffix $\langle e_{i+2}, \dots, e_n \rangle$ is a deterministic algorithm.
  This shows that concurrent operations commute.
\end{proof}

\subsection{Satisfying the Strong List Specification}

\begin{lemma}\label{state-correspondence}
  Let $G$ be a valid event graph, let $\mathit{doc} = \mathsf{replay}(G)$ be the document state resulting from replaying $G$, and let $S$ be the internal state after replaying $G$.
  Then the $i$th element in $\mathit{doc}$ corresponds to the $i$th record with $s_e = \texttt{Ins}$ in the internal state (counting placeholders as having $s_e = \texttt{Ins}$, and not counting records with $s_e = \texttt{Del}$).
  Moreover, the set of elements in $\mathit{doc}$ is exactly the elements that have been inserted but not deleted in $G$:
  \begin{multline*}
    (\exists i \in [0, n-1]: \mathit{doc}[i] = (\mathit{id}, c)) \Longleftrightarrow \\
    (\exists a \in G, i \in \mathbb{N}: a.\mathit{id} = \mathit{id} \wedge a.\mathit{op} = \mathit{Insert}(i,c)) \;\wedge \\
    (\nexists b \in G, i \in \mathbb{N}: b.\mathit{op} = \mathit{Delete}(i) \;\wedge \\
    \mathsf{replay}(\mathsf{Events}(b.\mathit{parents}))[i] = (\mathit{id}, c)).
  \end{multline*}
\end{lemma}
\begin{proof}
  Let $E = \langle e_1, e_2, \dots, e_n \rangle$ be some topological sort of $G$, and assume that we replay $G$ in this order.
  By Lemma \ref{lemma-deterministic} it does not matter which of the possible orders we choose.
  We then prove the thesis by induction over $n$, the number of events in $G$.
  The base case is trivial: $G=\{\}$, $\mathit{doc}=\langle \rangle$, so there are no events, no records in the internal state, and no elements in the document state.

  Inductive step: Let $E_k = \langle e_1, e_2, ..., e_k \rangle$ with $k<n$ be a prefix of $E$.
  Since the set of events in $E_k$ also forms a valid event graph, we can assume the inductive hypothesis, namely that replaying $E_k$ results in a document corresponding to the records with $s_e = \texttt{Ins}$ in the resulting internal state, and the document contains exactly those elements that have been inserted but not deleted by an operation in $E_k$.
  We now add $e_{k+1}$, the next event in the sequence $E$, to the replay.
  We do this by transforming $e_{k+1}$ using the internal state obtained by replaying $E_k$, and applying the transformed operation to the document state from $E_k$.
  We need to show that the invariant is still preserved in the following two cases: either (1)~$e_{k+1}.\mathit{op} = \mathit{Insert}(j,c)$ for some $j$, $c$, or (2)~$e_{k+1}.\mathit{op} = \mathit{Delete}(j)$ for some $j$.
  We also have to consider the case where the internal state is cleared, but we begin with the case where no state clearing occurs.

  Case (1): The set of elements that have been inserted but not deleted grows by $(e_{k+1}.\mathit{id}, c)$ and otherwise stays unchanged.
  The transformation of an insertion operation is always another insertion operation.
  The document state is thus updated by inserting the same element $(e_{k+1}.\mathit{id}, c)$, and otherwise remains unchanged.
  Moreover, the transformed index of that insertion is computed by counting the number of internal state records with $s_e = \texttt{Ins}$ that appear before the new record in the internal state, and the new record also has $s_e = \texttt{Ins}$, and the $s_e$ property of no other record is updated, so the correspondence between internal state records and document state is preserved.

  Case (2): The element being deleted is located at index $j$ in the document at the time $e_{k+1}$ was generated, which is $\mathsf{replay}(\mathsf{Events}(e_{k+1}.\mathit{parents}))$.
  We compute this element by retreating and advancing events until the prepare version equals $e_{k+1}.\mathit{parents}$, and then finding the $j$th (zero-indexed) record that has $s_p = \texttt{Ins}$ in the internal state.
  The records with $s_p = \texttt{Ins}$ are those that have been inserted but not deleted in events that happened before $e_{k+1}$, and thus the $j$th such record corresponds to $\mathsf{replay}(\mathsf{Events}(e_{k+1}.\mathit{parents}))[j]$.
  Before applying $e_{k+1}$, this record may have either $s_e = \texttt{Ins}$ or $s_e = \texttt{Del}$.
  If $s_e = \texttt{Ins}$, we update it to $s_e = \texttt{Del}$, and transform $e_{k+1}$ into a deletion whose index is the number of $s_e = \texttt{Ins}$ to the left of the target record in the internal state; by the inductive hypothesis, this is the correct document element to be deleted.
  If $s_e = \texttt{Del}$ before applying $e_{k+1}$, that event is transformed into a no-op, since another operation in $E_k$ has already deleted the element in question from the document state.
  In either case, we preserve the invariants of the induction.

  If $e_{k+1}$ is a critical version, we clear the internal state and replace it with a placeholder.
  By the definition of critical version, every event in $E_k$ and $e_{k+1}$ happened before every event in the rest of $E$.
  Therefore, after retreating and advancing any event after $e_{k+1}$, any internal state record with $s_e = \texttt{Del}$ will also have $s_p = \texttt{Del}\; k$ for some $k>0$, and any internal state record with $s_e = \texttt{Ins}$ will also have $s_p = \texttt{Ins}$ unless it is deleted by an event after $e_{k+1}$.
  Since an internal state with $s_e = \texttt{Del}$ can never move to state $s_e = \texttt{Ins}$, this means that any records with $s_e = \texttt{Del}$ as of the critical version can be discarded, since they will never again be needed for transforming the index of an operation after $e_{k+1}$.
  Moreover, since all of the remaining records have $s_e = s_p = \texttt{Ins}$ as of the critical version, and since the replay of the remaining events in $E$ will never need to advance or retreat an event prior to the critical version, all of the records in the internal state can all be replaced by a single placeholder while still preserving the invariants of the induction.
\end{proof}

\begin{theorem}\label{main-theorem}
  The \algname algorithm satisfies the strong list specification (Definition \ref{strong-list-spec}).
\end{theorem}
\begin{proof}
  Given a valid event graph $G$, let $\mathsf{replay}(G)$ be the replay function based on \algname, as introduced earlier.
  We must show that there exists a list order $\mathit{lo} \subset \mathsf{ID} \times \mathsf{ID}$ that satisfies the conditions given in Definition \ref{strong-list-spec}.
  We claim that this list order corresponds exactly to the sequence of records and placeholders in the internal state after replaying the entire event graph $G$.
  By Lemma \ref{lemma-deterministic}, this internal state exists and is unique.
  This correspondence is more apparent if we assume a variant of \algname that does not clear the internal state on critical versions, but we also claim that performing the optimisations in \autoref{clearing} preserves this property.

  To begin, note that the internal state is a totally ordered sequence of records, and that (aside from clearing the internal state) we only ever modify this sequence by inserting records or by updating the $s_p$ and $s_e$ properties of existing records.
  Thus, if a record with ID $\mathit{id}_i$ appears before a record with ID $\mathit{id}_j$ at some point in the replay, the order of those IDs remains unchanged for the rest of the replay.
  We define the list order $\mathit{lo}$ to be the ordering relation among IDs in the internal state after replaying $G$ using a \algname variant that does not clear the internal state.
  This order exists, is unique (Lemma \ref{lemma-deterministic}), and is by definition transitive, irreflexive, and total, so it meets requirement (2) of Definition \ref{strong-list-spec}.

  Let $e \in G$ be any event in the graph, and let $G_e = \{e\} \cup \mathsf{Events}(e.\mathit{parents})$ be the subset of $G$ consisting of $e$ and all events that happened before $e$.
  Note that $G_e$ satisfies the conditions in Definition \ref{valid-graph}, so it is also valid.
  Let $\mathit{doc}_e = \mathsf{replay}(G_e) = \langle (\mathit{id}_0, c_0), \dots, (\mathit{id}_{n-1}, c_{n-1}) \rangle$ be the document state immediately after locally generating $e$.
  Since $\mathsf{replay}$ is deterministic (Lemma \ref{lemma-deterministic}), $\mathit{doc}_e$ exists and is unique.

  By Lemma \ref{state-correspondence}, $\mathit{doc}_e$ contains exactly the elements that have been inserted but not deleted in $G_e$, which is requirement (1a) of Definition \ref{strong-list-spec}.
  Also by Lemma \ref{state-correspondence}, the $i$th element in $\mathit{doc}_e$ corresponds to the $i$th record with $s_e = \texttt{Ins}$ in the internal state obtained by replaying $G_e$.
  Since any pair of IDs that are ordered by the internal state derived from $G_e$ retain the same ordering in the internal state derived from $G$, we know that the ordering of elements in $\mathit{doc}_e$ is consistent with the list order $\mathit{lo}$, satisfying requirement (1b) of Definition \ref{strong-list-spec}.

  Finally, to demonstrate requirement (1c) of Definition \ref{strong-list-spec} we assume that $e.\mathit{op} = \mathit{Insert}(i,c)$, and we need to show that $\mathit{doc}_e [i] = (e.\mathit{id}, c)$.
  Since $G_e$ contains only $e$ and events that happened before $e$, but no events concurrent with $e$, we know that immediately before applying $e$, every record in the internal state will have $s_p = \texttt{Ins}$ if and only if it has $s_e = \texttt{Ins}$ (because there are no events that are reflected in the effect version but not in the prepare version $e.\mathit{parents}$).
  Therefore, the set of records that are counted while mapping the original insertion index $i$ to an internal state record equals the set of records that are counted while mapping the internal record back to a transformed index.
  Thus, the transformed index of the insertion is also $i$, and therefore the new element is inserted at index $i$ of the document as required.
  This completes the proof that \algname satisfies the strong list specification.
\end{proof}
\fi 
\end{document}